\documentclass[journal,draftcls,onecolumn,a4paper,12pt]{IEEEtran}
\usepackage{bm}
\usepackage{color}
\usepackage{amsthm,bm}
\usepackage{amsfonts}
\usepackage{amsmath}
\usepackage{graphicx}      
\usepackage[english]{babel}
\usepackage[utf8]{inputenc}
\usepackage{algorithm}
\usepackage[noend]{algpseudocode}
\usepackage{ upgreek }
\usepackage{multirow}
\usepackage{bbm}
\DeclareMathAlphabet{\mathpzc}{OT1}{pzc}{m}{it}
\usepackage{float}
\usepackage{amsmath,amsfonts,amssymb,amsthm}
\usepackage{algorithm}
\usepackage[noend]{algpseudocode}
\usepackage{chngcntr}
\usepackage{graphicx}
\usepackage{caption}
\usepackage{subcaption}
\usepackage{todonotes}
\usepackage{tabularx}
\usepackage{accents}
\usepackage{enumerate}
\usepackage[shortlabels]{enumitem}
\DeclareMathOperator*{\argmin}{argmin}  
\newcommand\thickbar[1]{\accentset{\rule{.4em}{.8pt}}{#1}}
\usepackage{ocg-p}
\newcommand{\mato}[2]{\left[\begin{array}{#1} #2 \end{array}\right]}
\newtheorem{thm}{Theorem}{}
\newtheorem{lem}{Lemma}{}
{}
\newtheorem{rem}{Remark}{}
\newcommand{\TODO}[1]{
	\newsavebox\todoboxa
	\savebox\todoboxa{\begin{ocg}[printocg=never]{ToDo screen}{todo4screen}{1}
			\todo[inline]{#1}
	\end{ocg}}
	\newsavebox\todoboxb
	\savebox\todoboxb{\begin{ocg}[printocg=always]{ToDo print}{todo4print}{0}
			\todo[inline, color=white]{#1}
	\end{ocg}}
	\noindent\usebox\todoboxb\llap{\usebox\todoboxa}
	\global\let\todoboxa\relax
	\global\let\todoboxb\relax
}

{}

\theoremstyle{plain}
\newtheorem{assumption}{Assumption}

\makeatletter
\def\user@resume{resume}
\def\user@intermezzo{intermezzo}
\newcounter{previousequation}
\newcounter{lastsubequation}
\newcounter{savedparentequation}
\renewenvironment{subequations}[1][]{%
	\def\user@decides{#1}%
	\setcounter{previousequation}{\value{equation}}%
	\ifx\user@decides\user@resume 
	\setcounter{equation}{\value{savedparentequation}}%
	\else  
	\ifx\user@decides\user@intermezzo
	\refstepcounter{equation}%
	\else
	\setcounter{lastsubequation}{0}%
	\refstepcounter{equation}%
	\fi\fi
	\protected@edef\theHparentequation{%
		\@ifundefined {theHequation}\theequation \theHequation}%
	\protected@edef\theparentequation{\theequation}%
	\setcounter{parentequation}{\value{equation}}%
	\ifx\user@decides\user@resume 
	\setcounter{equation}{\value{lastsubequation}}%
	\else
	\setcounter{equation}{0}%
	\fi
	\def\theequation  {\theparentequation  \alph{equation}}%
	\def\theHequation {\theHparentequation \alph{equation}}%
	\ignorespaces
}{%
	\ifx\user@decides\user@resume
	\setcounter{lastsubequation}{\value{equation}}%
	\setcounter{equation}{\value{previousequation}}%
	\else
	\ifx\user@decides\user@intermezzo
	\setcounter{equation}{\value{parentequation}}%
	\else
	\setcounter{lastsubequation}{\value{equation}}%
	\setcounter{savedparentequation}{\value{parentequation}}%
	\setcounter{equation}{\value{parentequation}}%
	\fi\fi
	\ignorespacesafterend
}
\makeatother

\ifCLASSINFOpdf
\else
\fi
\begin{document}
%
\title{Distributed Dual Gradient Tracking for Priority-Considered Load Shedding}
%
%
%

\author{Ismi Rosyiana Fitri, Jung-Su Kim$ $
\thanks{Ismi Rosyiana Fitri and Jung-Su Kim are with the Department of Electrical and Information Engineering, Seoul National University of Science and Technology, Seoul 01811, Korea e-mail: jungsu@seoultech.ac.kr. This work has been submitted to the IEEE for possible publication. Copyright may be transferred without notice, after which this version may no longer be accessible.}}

%
%

\markboth{Journal of \LaTeX\ Class Files,~Vol.~14, No.~8, August~2015}%
{Shell \MakeLowercase{\textit{et al.}}: Bare Demo of IEEEtran.cls for IEEE Journals}
%



\maketitle

	\begin{abstract}
	This paper studies two fundamental problems in power systems: the economic dispatch problem (EDP) and load shedding. For the EDP, an extension of the problem considering the transmission losses is presented. Because the optimization problem is non-convex owing to quadratic equality constraints modeling the transmission losses, a convex relaxation method is presented.  	Under a particular assumption, it is shown that a distributed algorithm can be designed to solve the considered EDP.
	Furthermore, this work aims to handle an overloading problem by employing an optimal load shedding method. Emphasis is placed on scheduling the load shedding when there exist some priorities on the loads.  First, a novel optimization problem is proposed to obtain the load shedding according to the predefined priority order. Then, a distributed algorithm solving the optimization problem is presented.
	Finally, this paper presents how to integrate the proposed algorithms, i.e., the EDP and the load shedding, in a distributed manner. Simulation results are presented to demonstrate the effectiveness of the proposed method.
\end{abstract}

\section{Introduction}\label{sec1}
\subsection{Background and motivation}
There has been a growing interest in developing distributed energy resource (DER) systems to reduce gas emissions and mitigate global warming.
Over the last few years, the economic dispatch problem (EDP) has been studied to utilize and manage DER efficiently. The problem is formulated to allocate power from the generators such that the aggregate generation cost is minimized while achieving supply-demand balance \cite{fan_zhang_1998}. 

Traditionally, the EDP is solved in a centralized manner using global system-wide information. However, such an approach suffers from single-point failure; thus, it is inherently not suitable for distributed energy systems \cite{Qiang2016}. To deal with this problem, distributed control approaches are widely studied to make up for the weak point of the centralized controller \cite{Yang2019}. 

Owing to the uncertainty and varying characteristics of DERs, the available power from the DERs and energy storage may not be sufficient to meet the load requirements. Thus, load shedding may be necessary to avoid a total collapse of the electrical supply \cite{Lin2019}. The improved communication technologies and sensing capabilities of a distributed energy system give the operator a chance to handle this problem through demand-side management. For example, under a demand response (DR) program, the utilities may request that energy users curtail a portion or all of their load power demand in exchange for some incentives, discounts, or other concessions. In light of this, it is appropriate to consider varying load shedding scenarios due to various customers' choices on the offered load shedding program in real-time operation.

This paper aims to develop a distributed energy management that is applicable for ensuring the distributed energy systems' reliability and efficiency. We are particularly interested in solving two fundamental problems in power systems: the EDP and load shedding. This work attempts to allocate the optimal dispatched power through the EDP when transmission loss is considered. In addition, it aims to design an algorithm for scheduling the load shedding when there are loads that agree to be shed first due to some prior contract. A distributed energy management is then formulated so that the proposed algorithm for the load shedding can be helpful for ensuring the sufficiency of the available power to supply the remaining demand, meaning that the feasibility of the EDP is guaranteed.

\subsection{Literature review and problem descriptions}
\subsubsection{Distributed EDP}
Recent developments in distributed algorithms for the EDP have heightened the need for distributed agents to agree on the incremental costs \cite{Zhang2012}. In \cite{Kar2012, Yang2017, Wang2019}, consensus-based algorithms are developed particularly for a quadratic function representing the generation cost. In order to consider general convex cost functions, the algorithms are designed using the dual problem of the EDP, e.g., in \cite{Doan2017, Yangvarying2017, Xu2019}. Furthermore, it is shown in \cite{Nedic2018ImprovedCR} that the convergence of the existing algorithms can be improved when the generation costs are also smooth functions. 

The works mentioned above do not consider the effects of transmission losses that exist in practical implementation. 
The studies in \cite{Binetti2014} and \cite{Kim2019} have given an account of transmission losses in the EDP. The approach is able to yield the supply-demand balance by calculating the power loss at every iteration of the EDP's algorithm and then adding it to the constant demand.
In \cite{Xing2015}, the loss model in a quadratic function is fully accounted for obtaining the optimal dispatched power. Because of the complicated quadratic function modeling the power losses, the study considers an algorithm which includes an inner loop of consensus steps per iteration. It is well known that an inner loop places a more considerable communication burden and needs extra coordination among agents to terminate the inner loop. In \cite{Zhao2017, Lee2019}, the loss model is simplified. Hence, the EDP can be solved in a distributed manner, and the algorithm does not require any inner loop steps. Given the existing works above, this paper is interested in studying an extension of the traditional EDP that considers the quadratic function loss model but still allows us to solve the problem without any need for a consensus inner loop.

\subsubsection{Priority-considered load shedding}
To date, several studies have examined a suitable scheme for scheduling the load shedding, e.g., \cite{Faranda2007, Hussain2020}. Heuristically, the priority-based load shedding can be obtained by gradually allocating the load shedding from the highest to the lowest-priority customers \cite{Kato2014, Hussain2018}.
In \cite{Grewal1998, Teshome2015}, the scheduled load shedding is determined by minimizing some convex functions representing the cost or the adversity of the load shedding. Larger weights are assigned to high-priority customers so that the resulting load shedding is in accordance with the pre-selected priority order. In response to the time-varying energy market and new forecast data, the desired priority list or the buses' maximum allowed load shed might change. As a result, the pre-selected weights may no longer work; thus, they must be re-tuned, which is non-trivial in real-time operation. This work aims to design a distributed algorithm for a priority-considered load shedding that minimizes the need for tuning parameters
	in order to handle the changes in the desired load shedding scenario.

\subsection{Problem description and statement of contribution}
The main contributions of this paper are as follows: 
\begin{itemize}
	\item In this paper, the EDP is formulated such that transmission losses are taken into account.
	The proposed EDP is a non-convex optimization problem because of a quadratic equality constraint denoting the supply-demand balance requirement under consideration of transmission losses. This paper shows that a convex relaxation method can be used to obtain the solution of the formulated non-convex optimization. 
	\item Under some assumptions, it is shown that the EDP can be solved in a distributed manner such that the method does not require any inner loop algorithms.
	\item This paper proposes a novel optimization for the priority-considered load shedding such that it does not require a re-tuning procedure to deal with possibly time-varying priority orders. A distributed algorithm to solve the proposed optimization problem is presented. 
\end{itemize}
Moreover, several supporting consensus algorithms are presented for the purpose of integrating the two proposed algorithms, i.e., the EDP and the load shedding, in a distributed manner. Note that when an overloading condition occurs, an EDP is not feasible due to the lack of available power. Thus, the load shedding can be applied to restore the feasibility of the EDP.

The paper is constructed as follows. Section II describes the problem definitions of the EDP and the priority-considered load shedding. Section III and IV present the algorithm development for the EDP and the priority-considered load shedding, respectively. Section V discussed the necessary additional algorithm for integrating both proposed methods. The proposed algorithms are validated through numerical simulations, which are given in Section VI.

\subsubsection{Notation} We use boldface to distinguish vector $\bm{x}$ in $\mathbb{R}^n$ from the scalar $x$ in $\mathbb{R}$. Given a vector  $\bm{x}$, we denote $\|\bm{x}\|$ as its Euclidean norm, that is $\|\bm{x}\|=\sqrt{\bm{x}^T \bm{x}}$. Given matrix $A\in \mathbb{R}^{n\times n}$, $A\geq 0~(A>0)$ is used to denote a positive semi-definite (definite) matrix $A$. The cardinality of a set $\mathcal{M}$ is denoted by $|\mathcal{M}|$. Consider a set $\mathcal{X}$ defined as $\mathcal{X}=\{x\in \mathbb{R}|x^{min}\leq x\leq x^{max}\}$. We use $\mbox{relint}(\mathcal{X})$ to denote the relative interior of set $\mathcal{C}$ given by $\mbox{relint}(\mathcal{X})=\{x\in \mathbb{R}|x^{min}< x< x^{max}\}$ for any $x^{min}$ and $x^{max}$ satisfying $x^{min}<x^{max}$. In this paper, an overbar is used to denote the optimal solution of an optimization problem. For example, $\thickbar{x}$ represents the optimal solution
of $\min_{x}F(x)$. In order to avoid a heavy notation, the superscript is used in this paper. For instance, $\eta^{(1)}, \eta^{(2)}$, and $\eta^{(3)}$ represent three different variables. For scalars $x_1,x_2,\ldots,x_n$, $\bm{x}=[x_i]$, $i\in \{1,\ldots,n\}$, is the aggregate vector, i.e., $\bm{x}=[x_1,x_2,\ldots,x_n]^T$. 
Hence,  $\bm{x}=[x^{(i)}]$, $i\in \{1,\ldots,n\}$, is equivalent to $\bm{x}=[x^{(1)},x^{(2)},\ldots,x^{(n)}]^T$.

\section{Problem Definition and Assumptions}
Consider the power network with $n$ buses consisting of a generator or (and) a load where each bus is incorporated with a controller, that is an agent, such that each agent is able to communicate with at least one other agent. The communication topology is modeled as an undirected graph  $\mathcal{G}=(\mathcal{V},\mathcal{E})$, where $\mathcal{V}:=\{1,\ldots,n\}$. Denoting $\mathcal{N}_i$ as the set of neighbors of the $i$th agent, $j\in\mathcal{N}_i$ implies $i\in\mathcal{N}_j$.  This paper aims to develop a distributed algorithm for scheduling the generation or (and) the load shedding for each bus.  For this purpose, the agents are allowed to exchange their local information with their neighboring agents through communication network $\mathcal{G}$ satisfying the following assumption. 
\begin{assumption}\label{asm:graph}
	$\mathcal{G}$ is an undirected and connected graph \cite{lewis_2014}.
\end{assumption}

\subsection{EDP considering transmission losses}
Consider a set of buses $\mathcal{V}:=\{1,\ldots,n\}$, where $x_i$ and $d_i$, $i\in \mathcal{V}$, are the $i$th bus's dispatched power and forecast load demand, respectively. In this paper, the EDP is formulated such that power losses are taken into account.
\subsubsection{Transmission loss model}
Denoting $\bm{x}=[x_i]$, $i\in\mathcal{V}$, in this paper, the power losses are modeled as follows
	\begin{eqnarray}\label{eq:loss}
		\Upsilon(\bm{x})=\sum_{i\in\mathcal{V}}\sum_{j\in\mathcal{V}}x_ib_{ij}x_j+\sum_{i\in\mathcal{V}}b_{0i}x_i+b_{00},
	\end{eqnarray}where $B=[b_{ij}]\in\mathbb{R}^{n\times n}$ and $B=B^T$. Without loss of generality and for simplicity, $b_{0i}=b_{00}=0,i=\{1,\ldots,n\}$, is used hereafter. 
	\begin{assumption}\label{asm:Bpos}
		$B\geq 0$.
	\end{assumption}
	Note that under this assumption, there exists a real positive semi-definite matrix $ R=[r_{ij}]\in\mathbb{R}^{n\times n}$ such that $R^TR=B$.
\begin{assumption}\label{asm:global}
	The $i$th agent knows the $i$th column of $R$.
\end{assumption}
\begin{rem}
	The $b$ terms (i.e., $b_{ij},b_{0i}$, and $b_{00}$) in \eqref{eq:loss} are called loss-coefficient which mostly can only be computed by the central controller as extensive amount of global information is required. For example, in \cite[Ch. 13.3]{grainger_stevenson_chang_2016},
	the power losses are computed according to the parameters of transmission lines.
	For the purpose of solving an EDP problem considering the transmission loss in a distributed manner, the operator inevitably has to broadcast the necessary information regarding the loss model to the agents. For instance, the studies in \cite{Xing2015,Kim2019} require the operator to broadcast the value of $b_{i1},b_{i2},\ldots,b_{in}$ to the $i$th agent since $b_{ij}$ can not be computed locally by any agent (including the $i$th agent). In Assumption \ref{asm:global}, we assume the $i$th agent receives the information of the $i$th column of $R$.
	In the light of distributed systems, such assumptions make the algorithm depend on the information from the operator. Nevertheless, since the loss coefficient do not change frequently, it is sufficient for the grid operator to send the information to the agents from time to time \cite[Ch. 13.3]{grainger_stevenson_chang_2016}.
\end{rem}
\subsubsection{EDP formulation}
	This work aims to solve the following problem
	\begin{subequations}\label{edp2}
		\begin{eqnarray}
			\min_{\substack{\bm{x}=[x_i],x_i\in\mathcal{X}_i\\i=1,\ldots,n}}  &\sum_{i\in \mathcal{V}}^{}C_i(x_i),\\
			\mbox{subject to}&\Upsilon(\bm{x})+\sum\limits_{i\in \mathcal{V}}d_i -\sum\limits_{i\in \mathcal{V}}x_i=0,\label{eqbal}
		\end{eqnarray}
\end{subequations}
where $\Upsilon(\mathbf{x})$ is given in \eqref{eq:loss}, $C_i(\cdot)$ denotes the generation cost of the $i$th generator, and 
\begin{eqnarray*}
	\mathcal{X}_i:=\{x_i|x_i^{min}\leq x_i \leq x_i^{max}\}.
\end{eqnarray*}
For the buses with no generator, $x_i^{min}=x_i^{max}=0$ and $C_i=0$, otherwise $0\leq x_i^{min}<x_i^{max}$. Moreover, if the $i$th bus has no load, then $d_i=0$. Note that \eqref{eqbal} represents the supply-demand balance constraint considering the transmission losses.
\subsubsection{Assumptions for the EDP}The following describes the assumptions considered in this paper.
\begin{assumption}\label{asm:conv} For all $i\in\mathcal{V}$, the function $C_i:\mathbb{R}\rightarrow\mathbb{R}$ is proper, closed, and $c_i$-strongly convex,  i.e., for all $x,y\in\mathbb{R}$ we have 
		\begin{eqnarray*}
			C_i(y)\geq  C_i(x)+\nabla C_i(x)\cdot(y\!-\!x)+\frac{c_i}{2}(y-x)^2.
		\end{eqnarray*}
		Moreover, $\nabla C_i(x_i)\geq0$ for all $x_i\in\mathcal{X}_i$ and $\nabla C_i(x_i)>0$ for all $x_i \in \mbox{relint}(\mathcal{X}_i)$.
\end{assumption}
\begin{assumption}\label{asm:feas}
	Denoting $\bm{x}^{min}\!=\![x_i^{min}]$, $\bm{x}^{max}\!=\![x_i^{max}]$, for all $i\in \mathcal{V}$, it follows that
	\begin{enumerate}[start=1,label={\bfseries A5.\arabic*},leftmargin=*]
		\item $\sum_{i\in \mathcal{V}}x_i^{min}<\sum_{i\in \mathcal{V}}d_i+\Upsilon(\bm{x}^{min})$,
		\item $\Upsilon(\bm{x}^{min})\! \leq\! \sum_{i\in \mathcal{V}}(x_i^{min})$, $\Upsilon(\bm{x}^{max})\! \leq\! \sum_{i\in \mathcal{V}}(x_i^{max})$,
	\end{enumerate}
	where $0\leq x_i^{min}\leq x_i^{max}$ and $d_i\geq 0$.
\end{assumption}

{Note that it is easy to satisfy Assumption \ref{asm:conv}.} For instance, it is widely known that the generation cost can be modeled as $C(x)=\mathsf{a}x^2+\mathsf{b}x+\mathsf{c}$ with $\mathsf{a}_i\geq0,\mathsf{b}_i\geq 0$, $\mathsf{c}\in\mathbb{R}$, and $x> 0$. 
Moreover, Assumption {A5.1} indicates that the required supply-demand balance \eqref{eqbal} is feasible under the minimum-capacity constraint, and {A5.2} means that the total transmission losses are smaller than the total available power. 
Under the assumption that \eqref{edp2} is feasible, we aim to solve the problem using a fully distributed algorithm. The next section presents the problem formulation for the load shedding. Note that the load shedding can be applied for handling the case when
an overloading condition occurs.

\subsection{Priority-considered load shedding}
Let $y_{tot}$ be the total amount of load shedding to be scheduled in the load buses and $y_i$ be the shed load at the $i$th bus such that $\sum_{i=1}^{n}y_i=y_{tot}$.  Denoting $y_i^{max}$ as the maximum allowed load shedding at the $i$th bus, it follows that $0\leq y_i\leq y_i^{max}$. In this paper, it is assumed that the load buses are divided into the following categories.
\begin{itemize}
	\item Prioritized shedding buses ($i\in \mathcal{M}$): The $i$th bus in this category is shed according to its predefined priority order $p_i$, where $p_i$ denotes the priority order for the $i$th bus and the set $\mathcal{M}$ is the set of buses belonging to this category.  
	\item Regular load buses ($i\!\in\! \mathcal{V}\backslash\mathcal{M}$):  Upon shedding the buses belonging to this set, 
	the customers suffer from not using their appliances and $D_i(y_i)$ represents the value of damages or loss due to the shed load $y_i$. Moreover, to alleviate the pain of load shedding, each customer receives an incentive payment  $\mathsf{r}_iy_i$, $\mathsf{r}_i>0$. 
\end{itemize}

\subsubsection{\textbf{LS (Load Shedding) Problem}}In this work, we aim to design a distributed algorithm for scheduling the load shedding such that the following conditions are met:
	\begin{enumerate}
		\item It is desirable that the load shedding is determined in the buses belonging to the set $\mathcal{M}$ first, before cutting the loads in  $\mathcal{V}\backslash\mathcal{M}$.
		\item Among the buses belonging to the set $\mathcal{M}$, $y_i$ is selected according to predefined priority $p_i$ without knowing the priority order of other customers.
		\item Due to the DR, the customer in the set  $\mathcal{V}\backslash\mathcal{M}$ aims to maximize its utility represented by $\mathsf{r}_iy_i-D_i(y_i)$ \cite{Nali2011}. In other words, $y_i$ is determined such that $\sum_{i\in\mathcal{V\backslash M}}D_i(y_i)-\mathsf{r}_iy_i$ is minimized. 
\end{enumerate}
\subsubsection{Notation for the priority-considered load shedding}

Given the desired priority order $p_i$ for each $i$th bus, $i\in \mathcal{M}$, suppose that the lowest order of priority is $m$, that is, $\max_i p_i=:m$. Introduce the set $\mathcal{M}_\ell$  given by $\mathcal{M}_\ell:=\{i|p_i=\ell\},  \ell=1,2,\ldots,m$.
Then, $\mathcal{M}$ is their union, i.e., $\mathcal{M}=\mathcal{M}_1\cup\ldots\cup\mathcal{M}_m$.

\subsubsection{Assumption for the load shedding}
\begin{assumption}\label{asm:conv2}For all $i\in\mathcal{V}$, the function $D_i:\mathbb{R}\rightarrow\mathbb{R}$ is proper, closed, and strongly convex.
\end{assumption}
In \cite{Nali2011}, a quadratic function is used to model the load shedding cost $D_i(y_i)$.

\section{Proposed EDP Considering Power Losses}
This section is entirely dedicated to present the solution of the considered EDP described in the previous section. Let us assume that the EDP \eqref{edp2} is feasible and that the following holds.
\begin{assumption}\label{asm:slater}[\textit{Slater's condition}]
	There exists a point $(\tilde{x}_1,\ldots,\tilde{x}_n)$ that lies in the interior of $\mathcal{X}_1\times\mathcal{X}_2\times\ldots\times \mathcal{X}_n$ such that $\Upsilon(\tilde{\bm{x}})+\sum_{i\in \mathcal{V}}d_i -\sum_{i\in \mathcal{V}}\tilde{x}_i = 0$ where $\tilde{\bm{x}}=[\tilde{x}_i]$.
\end{assumption}
\subsection{Development of a distributed algorithm for the EDP}
This section describes several equivalent optimization problems that are helpful for solving \eqref{edp2} in a distributed manner.
\subsubsection{An equivalent optimization problem for the EDP}Let us introduce a variable $u_i\in\mathbb{R}$, $i\in\mathcal{V}$. Note that the following optimization problem has the same solution as  \eqref{edp2}.
	\begin{subequations}\label{edp2equiv}
		\begin{eqnarray}
			\min_{\substack{\bm{x}=[x_i],\bm{u}=[u_i],\\x_i\in\mathcal{X}_i,\forall i\in\mathcal{V}}} &\sum_{i\in \mathcal{V}}^{}C_i(x_i),\\
			\mbox{subject to }&\sum\limits_{i\in \mathcal{V}}u_i^2+\sum\limits_{i\in \mathcal{V}}d_i -\sum\limits_{i\in \mathcal{V}}x_i=0,\label{eqbalequiv}\\
			&R\bm{x}=\bm{u},
		\end{eqnarray}
	\end{subequations}
	where $R^TR=B$. The optimization problem is defined such that a new slack variable $u_i$ is added and the loss model $\bm{x}^TB\bm{x}$ is substituted by $\bm{u}^T\bm{u}$ where $\bm{u}=R\bm{x}$. Since the feasible set and the objective function of \eqref{edp2equiv} are the same as those of \eqref{edp2}, the optimal solution of \eqref{edp2} can be found by solving \eqref{edp2equiv}. 
\subsubsection{Relaxation method for the non-convex EDP}
In \eqref{edp2equiv}, the problem is non-convex due to the quadratic equality constraint \eqref{eqbalequiv}; thus, obtaining the optimal solution can be non-trivial. To deal with this problem, consider the following optimization problem
	\begin{subequations}\label{edp2rel}
		\begin{eqnarray}
			\min_{\substack{\bm{x}=[x_i],\bm{u}=[u_i],\\x_i\in\mathcal{X}_i,\forall i\in\mathcal{V}}} &\sum_{i\in \mathcal{V}}^{}C_i(x_i),\\
			\mbox{subject to }&\sum\limits_{i\in \mathcal{V}}u_i^2+\sum\limits_{i\in \mathcal{V}}d_i -\sum\limits_{i\in \mathcal{V}}x_i\leq0,\label{eqbalrel}\\
			&R\bm{x}=\bm{u}\label{pxu}.
		\end{eqnarray}
	\end{subequations}Note that this problem is convex since the quadratic equality in \eqref{eqbalequiv} is relaxed into an inequality. 
	Nevertheless, the following theorem shows that the optimal solution of the relaxed problem \eqref{edp2rel} is also that of the EDP \eqref{edp2equiv}, equivalently \eqref{edp2}. 
\begin{thm}\label{thm:relax} [\textbf{Relaxation}]
		Suppose that Assumptions \ref{asm:Bpos}, \ref{asm:conv}, \ref{asm:feas}, and \ref{asm:slater} hold. If $(\thickbar{x}_i,\thickbar{u}_i),i\in\mathcal{V},$ is the optimal solution of \eqref{edp2rel}, then, $(\thickbar{x}_i,\thickbar{u}_i),i\in\mathcal{V},$ is also the optimal solution of \eqref{edp2equiv}.
\end{thm}
\begin{proof}
	Let $r_{ij}$ be the element of $R$ on the $i$th row and the $j$th column, and the Lagrangian function of \eqref{edp2rel} given as
	\begin{eqnarray}\label{globallagr}
		\mathcal{L}(\bm{x},\bm{u},\lambda,\bm{\xi})\!\!\!&=&\!\!\!\sum_{i\in\mathcal{V}}C_i(x_i)+\lambda\sum\limits_{i\in \mathcal{V}}\left(u_i^2+d_i -x_i\right)\nonumber\\&&+
		\sum_{j\in\mathcal{V}}\xi^{(j)}\left(-u_{j}+\sum\limits_{i\in \mathcal{V}}r_{ji}x_i\right),
	\end{eqnarray}
	where $x_i\in\mathcal{X}_i,\bm{x}=[x_i],\bm{u}=[u_i]$, $i\in\mathcal{V},$ $\bm{\xi}=[\xi^{(j)}]$, $j\in\mathcal{V},$ and $\lambda\geq 0$. Denote $(\thickbar{x}_i,\thickbar{u}_i)$ and $(\thickbar{\lambda},\thickbar{\bm{\xi}})$ as the primal and dual optimal solution of \eqref{edp2rel}, respectively. Since \eqref{edp2rel} is convex and the Slater's condition holds, the KKT conditions give the necessary and sufficient condition for the optimal solution of \eqref{edp2rel}. In this paper, the theorem is proven by showing that the following holds.
	\begin{subequations}\label{kkt}
		\begin{eqnarray}
			&\thickbar{x}_i\in\mathcal{X}_i,\quad\!(\thickbar{x}_i,\thickbar{u}_i)=\argmin\limits_{x_i\in\mathcal{X}_i,u_i}\mathcal{L}(\bm{x},\bm{u},\lambda,\bm{\xi}),i\in\mathcal{V},\\
			&\thickbar{\lambda}>0,\sum_{i\in \mathcal{V}}\thickbar{u}_i^2+d_i -\thickbar{x}_i=0,\label{kkt2}\\
			&\sum_{i\in \mathcal{V}}r_{ji}\thickbar{x}_i=\thickbar{u}_j,~j=1,\ldots,n.\label{kkt3}
		\end{eqnarray}
	\end{subequations}
	Observe that the equality in \eqref{kkt2} shows that the optimal solution of \eqref{edp2rel} is also the optimal solution of \eqref{edp2equiv}. The optimality conditions given in  \eqref{kkt} is proven in the Appendix.  The proof shows that the monotonicity of the objective functions $C_i(x_i)$ on $\mbox{relint}(X_i)$ plays a key role in obtaining the result.
\end{proof}
In light of Theorem \ref{thm:relax}, the optimal solution \eqref{edp2equiv} can be obtained through the relaxed problem \eqref{edp2rel}.

\subsubsection{Utilizing a redundant constraint for the distributed EDP}
Consider the following optimization problem
	\begin{eqnarray}\label{edp2relmod}
		\min_{\substack{\bm{x}=[x_i],\bm{u}=[u_i],\\x_i\in\mathcal{X}_i,\forall i\in\mathcal{V}}} &\sum_{i\in \mathcal{V}}^{}C_i(x_i),\\
		\mbox{subject to }&\eqref{eqbalrel},~\eqref{pxu},u_i\in\mathcal{U},~i\in\mathcal{V},\nonumber
	\end{eqnarray}
	where $\mathcal{U}:=\{u|-u^{max}\leq u\leq u^{max}\}$ and 
	\begin{eqnarray}\label{setU}
		u^{max}\!=\! \sum_{i\in\mathcal{V}}(r_{i}^{max}x_i^{max}),
	\end{eqnarray}
	and $ r_i^{max}:=\max_{j=1,\ldots,n} |r_{ji}| $.
Observe that \eqref{edp2relmod} is formulated by adding a new constraint $u_i\in\mathcal{U}$ to \eqref{edp2rel}, and that the following holds
\begin{eqnarray*}
	&\!\!\!\!\!\!\!\!\! \sum_{j\in\mathcal{V}}r_{ij}x_j^{max}\leq u^{max},~\forall i,j\in\mathcal{V}.
\end{eqnarray*}
Moreover, since $x_i^{min}\leq x_i^{max}$, $-u^{max}\leq \sum_{j\in\mathcal{V}}r_{ij}x_j^{min}$, the constraint $u_i\in\mathcal{U}$ is always satisfied as long as $x_i$ and $u_i$ satisfy the constraints in \eqref{edp2rel}: $x_i\in\mathcal{X}_i$ and \eqref{pxu}. In other words, the new constraint $u_i\in\mathcal{U}$ does not alter the feasible region \eqref{edp2relmod}. As a result, the optimal solution of \eqref{edp2relmod} is equivalent to that of \eqref{edp2rel}. In view this, the optimal solution of \eqref{edp2rel} can be obtained by solving \eqref{edp2relmod}. Clearly, it is redundant to specify $u_i\in\mathcal{U}$ given that $x_i\in\mathcal{X}_i$ and \eqref{pxu} hold. However, it is shown later (c.f. Theorem \ref{thm:edp}) that the constraint $u_i\in\mathcal{U}$ is useful for designing a distributed algorithm to obtain the optimal solution of \eqref{edp2relmod}.
\subsubsection{Dual problem}
Let $\thickbar{\mathcal{S}}$ be the set of the optimal solutions of \eqref{edp2relmod}. In view of $\mathcal{X}_i$,
the feasible region of \eqref{edp2relmod} is closed and bounded. From Assumption \ref{asm:conv}, the objective function is convex and continuous. Thus, according to the Weierstrass theorem \cite{Bertsekas1989}, there exists an optimal point $(\thickbar{\bm{x}},\thickbar{\bm{u}})\in\thickbar{\mathcal{S}}$ that achieves  a global optimum. Moreover, under the Slater's condition given in Assumption \ref{asm:slater}, the strong duality holds and the set of the dual optimal solutions is nonempty. Hence, an algorithm based on a primal-dual method can be used for obtaining the optimal solution of \eqref{edp2relmod}. To this end, consider $\lambda_i,\xi_i^{(1)},\ldots,\xi_i^{(n)}$ as the $i$th agent's estimates of the Lagrangian multipliers $\lambda,\xi^{(1)},\ldots,\xi^{(n)}$, and the following local Lagrangian function
	\begin{eqnarray}\label{locallagr}
		\mathcal{L}_i(x_i,u_i,\lambda_i,\bm{\xi}_i)\!\!\!&=&\!\!\!C_i(x_i)+\lambda_i(u_i^2+d_i -x_i)+\xi_i^{(i)}(r_{ii}x_i-u_{i})+\!\!\!\!\sum_{j\in\mathcal{V},j\not=i}\xi_i^{(j)}r_{ji}x_i,\nonumber
\end{eqnarray}\noindent where $\bm{\xi}_i=[\xi_i^{(1)},\ldots,\xi_i^{(n)}]^T$, $x_i\in\mathcal{X}_i$, and $u_i\in\mathcal{U}$. Observe that $\sum_{i\in\mathcal{V}}\mathcal{L}_i=\mathcal{L}$ if and only if  $\bm{\xi}_i=\bm{\xi}_j$ and $\lambda_i=\lambda_j$ hold for all $i,j\in\mathcal{V}$.
	Denoting $h_i(\lambda_i,\bm{\xi}_i)=\inf_{x_i\in\mathcal{X}_i,u_i\in\mathcal{U}}\mathcal{L}_i(x_i,u_i,\lambda_i,\bm{\xi}_i)$, the dual problem of \eqref{edp2relmod} can be formulated as follows:
	\begin{subequations}\label{dualprob2}
		\begin{eqnarray}
			\min\limits_{\substack{\lambda_i\geq 0,\bm{\xi}_i,\\i=1,\ldots,n}}&& \sum_{i\in\mathcal{V}}q_i(\lambda_i,\bm{\xi}_i)\\
			\mbox{subject to }&&\lambda_{i}=\lambda_{j},~~\forall i,j\in\mathcal{V},\\
			&&\bm{\xi}_i=\bm{\xi}_j,~~\forall i,j\in\mathcal{V},
		\end{eqnarray}
	\end{subequations}
	where $q_i(\lambda_i,\bm{\xi}_i):=-h_i(\lambda_i,\bm{\xi}_i)$. Note that the objective function of \eqref{dualprob2} is a summation of several convex functions.  
	Utilizing this structure, the focus of the problem is placed on designing a distributed algorithm for \eqref{dualprob2}.

\subsection{Proposed distributed algorithm for the EDP}
In this section, the algorithm for solving the dual problem \eqref{dualprob2} is designed based on the distributed subgradient method \cite{Nedic2010}.  
In addition to $x_i,u_i,\lambda_i,$ and $\bm{\xi}_i$, let us introduce two variables $v_i\in\mathbb{R}$ and $\bm{w}_i\in\mathbb{R}^{n}$, $i\in\mathcal{V}$. The sequences of $x_i(k),u_i(k)$ for $k\geq0$ are defined as follows
	\begin{subequations}\label{alg_edp}
		\begin{eqnarray}
			\!\!v_i(k)\!\!\!\!\!&=&\!\!\!\!\!\sum\limits_{j\in\mathcal{V}}a_{ij}\lambda_j(k),\label{alg_edp0}\\
			\!\!\bm{w}_i(k)\!\!\!\!\!&=&\!\!\!\!\!\sum\limits_{j\in\mathcal{V}}a_{ij}\bm{\xi}_j(k),\label{alg_edp1}
		\end{eqnarray}
	\end{subequations}
	\vspace{-0.5cm}
	\begin{subequations}[resume]
		\begin{eqnarray}
			\!\!\!	\!\!\!\!\!\!\!\!\!\!\!\!\!\!\!\!\!\!\!&&[x_i(k\!+\!1),u_i(k\!+\!1)]\!=\!\inf_{\substack{x_i\in\mathcal{X}_i,u_i\in\mathcal{U}}}\!\!\!\mathcal{L}_i(x_i,u_i,v_i(k),\bm{w}_i(k)),\label{alg_edp_prim}\\
			\!\!\!\!\!\!	\!\!\!\!\!\!\!\!\!\!\!\!\!\!\!\!&&\lambda_i(k\!+\!1)\!=\!\mathcal{P}_{\mathcal{R}}\!\left[v_i(k)\!+\!\alpha(k)\left(u_i^2(k\!+\!1)\!+\!d_i\! -\!x_i(k\!+\!1)\right)\right]\!\!,~~\label{alg_edp3}\\
			\!\!\!\!\!\!\!\!\!\!\!\!\!\!\!\!\!\!\!\!\!\!&&\bm{\xi}_i(k\!+\!1)\!=\!\bm{w}_i(k)+\alpha(k)\bm{g}_i(x_i(k+1),u_i(k+1)),\label{alg_edp4}
		\end{eqnarray}
\end{subequations}
where $\alpha(k)$ is a stepsize sequence, $\mathcal{P}_{\mathcal{R}}$ is a projection onto the set $\mathcal{R}=\{\lambda|\lambda\geq 0\}$, $\bm{g}_i(x_i,u_i):=[g_i^{(1)}(x_i,u_i),\ldots,g_i^{(n)}(x_i,u_i)]^T$, and
\begin{eqnarray}\label{gdef}
		g_i^{(j)}(x_i,u_i)=\left\{\!\!\begin{array}{cc}r_{ii}x_i-u_i	&\mbox{if }j=i,\\
			r_{ji}x_i&\mbox{otherwise}.
		\end{array}\right.
\end{eqnarray}\noindent Moreover, $A:=[a_{ij}]$ is a doubly stochastic matrix, i.e., $\sum_{i\in\mathcal{V}}a_{ij}=\sum_{j\in\mathcal{V}}a_{ji}=1$ with $a_{ii}>0$, such that
$a_{ij}>0$ if and only if $j\in\mathcal{N}_i$ otherwise $a_{ij}=0$

In view of \eqref{alg_edp0}-\eqref{alg_edp1}, the $i$th agent broadcasts $\lambda_{i}(k)$ and $\bm{\xi}_{i}(k)$ to its neighboring agents which are used to update $v_j(k)$ and $\bm{w}_j(k)$, $j\in \mathcal{N}_i$. In step \eqref{alg_edp3}-\eqref{alg_edp4}, $\lambda_i$ and $\bm{\xi}_{i}$ are updated by moving the current values along the subgradient of $q_i$ at $v_i(k)$ and $\bm{w}_i(k)$ with step size $\alpha(k)$. Projection $\mathcal{P}_{\mathcal{R}}[\cdot]$ is used to make $\lambda_{i}$ nonnegative. With Assumption \ref{asm:global} in mind,  the proposed algorithm \eqref{alg_edp} can be implemented in a distributed way if $u_{max}$ is known to all agents. Note that, from the definition of $u_{max}$ given in \eqref{setU}, an average consensus method can be applied to generate $u_{max}$ in a distributed manner before running \eqref{alg_edp}. Furthermore, the proposed method does not need any inner loop algorithm at each iteration $k$.

\subsection{Convergence analysis}
Let $(\thickbar{x}_i,\thickbar{u}_i)$ and $(\thickbar{\lambda},\thickbar{\bm{\xi}})$ be the optimal solution of \eqref{edp2relmod} and \eqref{dualprob2} respectively, $i\in\mathcal{V}$. The following shows the convergence of \eqref{alg_edp}.
	\begin{thm}\label{thm:edp}
		Suppose that the EDP \eqref{edp2} is feasible, and that Assumptions \ref{asm:graph}, \ref{asm:Bpos}, \ref{asm:conv}, \ref{asm:feas}, and \ref{asm:slater} hold true. Consider the sequences \eqref{alg_edp} with a step size $\alpha(k)$ satisfying $\sum_{k=0}^{\infty}\alpha(k)=\infty$ and $\sum_{k=0}^{\infty}\alpha^2(k)<\infty$.
		Then, for all $ i\in\mathcal{V}$,  $(\lambda_i(k),\bm{\xi}_i(k))\!\!\rightarrow\!\!(\thickbar{\lambda},\thickbar{\bm{\xi}})$ and $(x_i(k),u_i(k))\rightarrow(\thickbar{x}_i,\thickbar{u}_i)$ as $k\rightarrow \infty$.
\end{thm}
\begin{proof}
	The feasibility of \eqref{edp2} implies the feasibility of \eqref{edp2relmod}. Note that the dual problem \eqref{dualprob2} has the similar form to the optimization problem considered in \cite{Nedic2010}. {According to \cite[Proposition 4]{Nedic2010}, the convergence of $(\lambda_i(k),\bm{\xi}_i(k))$ to dual optimal point $(\thickbar{\lambda},\thickbar{\bm{\xi}})$ can be shown when the subgradient of $q_i$ is uniformly bounded.}
	Observe that $|x_i(k)|\leq x_i^{max}$ and $|u_i^{(j)}(k)|\leq u^{max}$ for all $k\geq 1$. Having this,  the subgradient of $q_i$ is uniformly bounded,  i.e., there exists scalar $\Gamma>0$ such that $|u_i(k)^2+d_i -x_i(k)|\leq \Gamma$ and  $|g_i^{(j)}(x_i(k),u_i(k))|\leq\Gamma$ for all $k\geq 1$, $i,j\in\{1,\ldots,n\}$. As a result,  $(\lambda_i(k),\bm{\xi}_i(k))\!\!\rightarrow\!\!(\thickbar{\lambda},\thickbar{\bm{\xi}})$ as $k\rightarrow \infty$.
	
	\noindent
	Since \eqref{kkt} holds for the primal-dual optimal solution of \eqref{edp2rel}, it also holds for the primal-dual optimal solution of \eqref{edp2relmod}. Given that $C_i$ is strongly convex and $\thickbar{\lambda}>0$,  $\mathcal{L}_i(x_i,u_i,\lambda_i,\bm{\xi}_i)$ 
	is a strongly convex function for fixed $\lambda_i=\thickbar{\lambda}$ and ${\bm{\xi}_i}=\thickbar{\bm{\xi}}$. In other words, there exists  $\beta_i>0$ satisfying 
	\begin{eqnarray}\label{sub2}
		&&\!\!\!\!\!\!\!\!\!\!\!\frac{\beta_i}{2}\left\|\mato{c}{\thickbar{x}_i\\\thickbar{u}_i}-\mato{c}{{x}_i(k+1)\\ u_i(k+1)}\right\|^2\\
		&&\!\!\!\!\!\!\!\!\!\!\leq\mathcal{L}_i(\thickbar{x}_i,\thickbar{u}_i,\thickbar{\lambda},\thickbar{\bm{\xi}})\!-\mathcal{L}_i(x_i(k\!+\!1),u_i(k\!+\!1),\thickbar{\lambda},\thickbar{\bm{\xi}})\nonumber\\\!\!\!\!\!\!&&\!\!\!\!\!\!-\nabla\mathcal{L}_i(x_i(k\!+\!1),u_i(k\!+\!1),\thickbar{\lambda},\thickbar{\bm{\xi}})^T\left(\mato{c}{\!\!\!\thickbar{x}_i\!\!\!\\\!\!\!\thickbar{u}_i\!\!\!}\!-\!\mato{c}{\!\!\!{x}_i(k\!+\!1)\!\!\!\\\!\!\! u_i(k\!+\!1)\!\!\!}\right)\!.~\nonumber
	\end{eqnarray}
	Given that $x_i(k+1)$ and $u_i(k+1)$ satisfy \eqref{alg_edp_prim},  observe that $\lim_{k\rightarrow\infty}(v_i(k),\bm{w}_i(k))=(\thickbar{\lambda},\thickbar{\bm{\xi}})$ implies $\nabla\mathcal{L}(x_i(k\!+\!1),u_i(k\!+\!1),\thickbar{\lambda},\thickbar{\bm{\xi}})\rightarrow\bm{0}$ as $k\rightarrow \infty$. With this and  \eqref{sub2} in mind, we have
	\begin{eqnarray}\label{uppereq}
		&&\!\!\!\!\!\!\!\!\!\!\lim\limits_{k\rightarrow\infty}\sum_{i\in \mathcal{V}}\frac{\beta_i}{2}\left\|\mato{c}{\thickbar{x}_i\\\thickbar{u}_i}-\mato{c}{{x}_i(k+1)\\ u_i(k+1)}\right\|^2\\
		&&\!\!\!\!\!\!\!\!\leq\lim\limits_{k\rightarrow\infty}\sum_{i\in \mathcal{V}}\left(\mathcal{L}_i(\thickbar{x}_i,\thickbar{u}_i,\thickbar{\lambda},\thickbar{\bm{\xi}})\!-\mathcal{L}_i(x_i(k\!+\!1),u_i(k\!+\!1),\thickbar{\lambda},\thickbar{\bm{\xi}})\right).\nonumber
	\end{eqnarray}
	\textbf{Claim 1:} The strong-duality of \eqref{edp2relmod} yields
	\begin{eqnarray}\label{sub1}
		\!\!\!\!\!\!\!\!\!\!\!\!&&\!\!\!\!\!\!\!\!\!\!\!\!\lim\limits_{k\rightarrow\infty}\!\sum_{i\in \mathcal{V}}\!\!\mathcal{L}_i(\thickbar{x}_i,\thickbar{u}_i,\thickbar{\lambda},\thickbar{\bm{\xi}})\!-\!\mathcal{L}_i(x_i(k\!+\!1),u_i(k\!+\!1),\thickbar{\lambda},\thickbar{\bm{\xi}})\!=\!0.
	\end{eqnarray}
	The proof of the claim is given in Appendix. In view of \eqref{uppereq} and \eqref{sub1}, $(x_i(k),u_i(k))\rightarrow(\thickbar{x}_i,\thickbar{u}_i)$ as $k\rightarrow \infty$.
\end{proof}

\section{Priority-considered Load Shedding}
In this section, at first, a convex optimization problem is proposed for addressing the desired load shedding result described in \textbf{LS Problem}. Then, a distributed algorithm for solving the proposed optimization problem is presented. 

\subsection{Proposed optimization problem}
	Let us assume that the total amount of load shedding $y_{tot}$ is known at least by the central operator, and that each $i$th agent has a knowledge on its average value
	$s_i$, i.e.,
	\begin{eqnarray*}
		s_i:= y_{tot}/n.
	\end{eqnarray*}
	Consider a new variable $z_i, i\in \mathcal{M},$ denoting a nonnegative slack variable which is used to maximize $y_i$ belonging to the set $\mathcal{M}$. Moreover, define
\begin{eqnarray*}
		~\mathcal{Z}_i:=\{z_i|0\leq z_i\leq n s_i\},\quad\mbox{and}\quad\mathcal{Y}_i:=\{y_i|0\leq y_i \leq y_i^{max}\}.
\end{eqnarray*}\noindent For all  $i\in\mathcal{M}_\ell,\ell=1,\ldots,m,$ consider the function $f_i(z_i,y_i)$ defined by
\begin{eqnarray*}
		f_i(z_i,y_i):=\kappa z_i^2+\left(y_i-\frac{ns_i}{p_i}\right)^2.
\end{eqnarray*}\noindent Recall that $p_i$ is the desired priority order of the $i$th bus, e.g., $p_i=1\Leftrightarrow i\in\mathcal{M}_1$.
Having this, the proposed optimization problem for the priority-considered load shedding is given as follows:	%
	\begin{subequations}\label{multi8}
		\begin{eqnarray}
			\min_{\substack{y_i,\forall i\in \mathcal{V}\\z_i,\forall i\in\mathcal{M}}}\!\!\!\!&&\!\!\!\!\!\sum\limits_{i\in\mathcal{V}\backslash\mathcal{M}}\!\!\!\!\!\left(D_i(y_i)-\mathsf{r}_iy_i\right)+\!\!\sum\limits_{i\in\mathcal{M}}\!\! f_i(z_i,y_i)~~~~\nonumber\\
			\!\!\!\!\!\!\!\!\!\!\!\!\!\!\!\!\!\!\mbox{subject to}&&
			\!\!\!\!\!\! ~y_i\in \mathcal{Y}_i,~ \forall i\in\mathcal{V},~z_i\in\mathcal{Z}_i,~\forall i\in\mathcal{M}\nonumber\\
			&&\!\!\!\!\!\!\!\!\!\mbox{if }m=0:\nonumber\\
			&&\!\!\sum_{i\in \mathcal{V}}s_i-\sum_{i\in \mathcal{V}}y_i=0,\label{sum_shed}\\
			&&\!\!\!\!\!\!\!\!\!\mbox{if }m\not=0:\nonumber\\
			&&\!\! \sum\limits_{i\in \mathcal{V}}s_i-\!\!\!\sum\limits_{i\in\mathcal{M}_1}\!\!\left(z_i\!+\!y_i\right)=0,\label{eqbal8}\\
			&&\!\!\!\!\!\! \sum\limits_{i\in\mathcal{M}_{\ell-1}}\!\!\!\!\!z_{i}\!-\!\!\!\sum\limits_{i\in\mathcal{M}_{\ell}}\!\!\!\left(z_i\!+\!y_i\right)\!=\!0,~   \!\forall \ell\!=\!2,\ldots,m,  \label{new8_or}\\
			&&\!\!\!\! \sum\limits_{i\in\mathcal{M}_{m}}\!\!\!\!z_{i}\!-\!\!\!\sum_{i\in \mathcal{V}\backslash\mathcal{M}}y_i=0,\label{const_last}
		\end{eqnarray}
	\end{subequations}
	where $\kappa\geq1$ is a tuning parameter. Note that the proposed optimization problem is convex and consists of $m+1$ equality constraints.
\subsubsection{Priority-considered load shedding through \eqref{multi8}}
The objective function consists of two terms. The first term $\sum_{i\in\mathcal{V}\backslash\mathcal{M}}D_i(y_i)-\mathsf{r}_iy_i$ aims to maximize the utility of the customers belonging to the set $\mathcal{V}\backslash\mathcal{M}$. Moreover, the following shows that $f_i(z_i,y_i)$ is used to make the optimal solution of \eqref{multi8} adhere to the desired priority list.
	\begin{enumerate}[label=(\roman*)]
		\item Substituting of \eqref{new8_or} and \eqref{const_last} into \eqref{eqbal8} yields \eqref{sum_shed}, i.e., $\sum_{i\in\mathcal{M}}{y_i}=y_{tot}$.
		\item With \eqref{eqbal8} in mind, minimizing $\kappa\sum_{i\in\mathcal{M}_1}\! z_i^2$ is equivalent to minimizing $\sum_{i\in \mathcal{V}}s_i-\sum_{i\in\mathcal{M}_1}\!y_i$. Hence, giving $\kappa>\nabla D_i(y_i)$ for all $i\in\mathcal{V}\backslash\mathcal{M}$ assigns more cost to $\kappa\sum_{i\in\mathcal{M}_1}\! z_i^2=\kappa \left(\sum_{i\in \mathcal{V}}s_i-\sum_{i\in\mathcal{M}_1}\!y_i\right)^2$ compared to $\sum_{i\in\mathcal{V}\backslash\mathcal{M}}D_i(y_i)$. This means sufficiently large $\kappa$ results in more load shedding on the buses belonging to the set $\mathcal{M}_1$ compared to other buses.
		\item Equality \eqref{new8_or} for $\ell=2$ can be written as $\sum_{i\in \mathcal{M}_2}z_i=\sum_{i\in \mathcal{M}_1}z_i-\sum_{i\in\mathcal{M}_2}\!y_i$. Substituting $\sum_{i\in \mathcal{M}_1}z_i$ with its definition given in \eqref{eqbal8} yields $\sum_{i\in \mathcal{M}_2}z_i=\sum_{i\in \mathcal{V}}s_i-\sum_{i\in\mathcal{M}_1\cup\mathcal{M}_2}\!y_i$. Hence, minimizing $\kappa\sum_{i\in\mathcal{M}_2}\! z_i^2$ with $\kappa>\nabla D_i(y_i)$ for all $i\in\mathcal{V}\backslash\mathcal{M}$ can be seen as an attempt to allocate more load shedding on the buses belonging to the set $\mathcal{M}_1\cup\mathcal{M}_2$. 
		\item Note that from the points (ii) and (iii),  $y_i,i\in\mathcal{M}_1,$ is prioritized through the minimization of $ \kappa\sum_{i\in\mathcal{M}_1}\! z_i^2$ and $ \kappa\sum_{i\in\mathcal{M}_2}\! z_i^2$. Meanwhile, $y_i,i\in\mathcal{M}_2,$ is prioritized through the minimization of $ \kappa\sum_{i\in\mathcal{M}_2}\! z_i^2$ only. In light of this, $y_i,i\in \mathcal{M}_1,$ is prioritized more than $y_i,i\in \mathcal{M}_2$.
		\item This reasoning can be extended to show that the constraint \eqref{new8_or} helps to generate $y_i$ according to its desired priority order $p_i$.
\end{enumerate}
In addition, the objective function $\left(y_i-{ns_i}/{p_i}\right)^2$ is used to make the term strongly convex. Hence, the optimal solution $\thickbar{y}_i,i\in\mathcal{M},$ can be obtained using a primal-dual approach (c.f. Theorem \ref{thm:ls}). By definition, the function minimizes the difference between $y_i$ and ${y_{tot}}/{\ell}$, $i\in\mathcal{M}_\ell,\ell=1,\ldots,m$.
\subsubsection{Feasibility of \eqref{multi8}}
The following describes the condition when \eqref{multi8} is feasible.
\begin{assumption}\label{asm:lsfeas}
	$\sum_{i\in \mathcal{V}}y_i^{max}>  \sum_{i\in \mathcal{V}}s_i=y_{tot}$.
\end{assumption}
Note that since the constraint set in \eqref{multi8} is compact and the objective function is continuous, the set of the optimal solution of \eqref{multi8} is
nonempty. Moreover, under Assumption \ref{asm:lsfeas}, there exists $\tilde{y}_1,\ldots,\tilde{y}_n$ that lie in the interior of $\mathcal{Y}_1\times\ldots\times \mathcal{Y}_n$ satisfying $\sum_{i\in\mathcal{V}}s_i-\tilde{y}_i=0$, meaning that the strong duality holds for \eqref{multi8}. 
\subsection{Distributed priority-considered load shedding}
Here, we present an algorithm to solve \eqref{multi8} in a distributed manner
\subsubsection{Dual Decomposition}
Let $\bm{y}=[y_i]$, ${i\in\mathcal{V}}$, and $\bm{z}=[z_i]$, ${i\in\mathcal{M}}$, consider the Lagrangian function of \eqref{multi8} given as
\begin{eqnarray}
		\mathcal{L}^{ls}(\bm{y},\bm{z},\eta^{(1)},\ldots,\eta^{(m+1)})=\sum\limits_{i\in\mathcal{V}\backslash\mathcal{M}}\!\!\!\!\!\left(\!D_i(y_i)\!-\!\mathsf{r}_iy_i\!\right)\!+\!\!\!\sum\limits_{i\in\mathcal{M}}\!\! f_i(z_i,y_i)+\sum_{i\in\mathcal{V}}\sum_{\nu=1}^{m+1}\eta^{(\nu)}\mathsf{g}_i^{(\nu)}(z_i,y_i),
\end{eqnarray}\noindent where $y_i\in\mathcal{Y}_i$ for all ${i\in\mathcal{V}}$ and $z_i\in\mathcal{Z}_i$ for all ${i\in\mathcal{M}}$. Moreover, with a slight abuse of notation,
\begin{subequations}
	\begin{eqnarray}
		\mathsf{g}_i^{(m+1)}(z_i,y_i)&=& \left\{ \begin{array}{cl}
			s_i-y_i
			& m=0, \\z_i & m\not=0,i\in\mathcal{M}_{m},
			\\-y_i & m\not=0,i\in\mathcal{V}\backslash\mathcal{M},
			\\0&\mbox{otherwise},
		\end{array}\right.\\
		\mathsf{g}_i^{(1)}(z_i,y_i)&=& \left\{ \begin{array}{cl}
			s_i-y_i-z_i
			& i\in \mathcal{M}_1 \\s_i & \mbox{otherwise} .
		\end{array}\right. \label{def_1}
	\end{eqnarray}
\end{subequations}
For $\nu=2,\ldots,m$, $\mathsf{g}_i^{(\nu)}(y_i,z_i)$  is defined as follows
\begin{subequations}[resume]
	\begin{eqnarray}
		\mathsf{g}_i^{(\nu)}(z_i,y_i)&=& \left\{ \begin{array}{cl}
			-y_i-z_i 
			& i\in \mathcal{M}_\ell,\nu=\ell \\z_i &  i\in \mathcal{M}_{\ell},\nu=\ell+1\\0& \mbox{otherwise} ,
		\end{array}\right. \label{def_fell}
	\end{eqnarray}
\end{subequations}
where $\ell\in\{1,\ldots,m\}$. Denoting $\eta_i^{(\nu)}$ as the $i$th agent's estimate of $\eta^{(\nu)}$, and $\bm{\eta}_i=[\eta_i^{(\nu)}]$, $\nu=1,\ldots,m+1$, consider the following local Lagrangian function
\begin{eqnarray}
\mathcal{L}^{ls}_i(y_i,z_i,\bm{\eta}_i):=F_i(y_i,z_i)+\sum_{\nu=1}^{m+1}\eta_i^{(\nu)}\mathsf{g}_i^{(\nu)}(z_i,y_i)
\end{eqnarray}
where, with a slight abuse of notation, $F_i(z_i,y_i)=D_i(y_i)-\mathsf{r}_iy_i$ for  $i\in\mathcal{V}\backslash\mathcal{M}$, and $F_i(z_i,y_i)=f_i(z_i,y_i)$ for $i\in\mathcal{M}$. Note that $\sum_{i\in\mathcal{V}}\mathcal{L}_i^{ls}=\mathcal{L}^{ls}$ if and only if  $\bm{\eta}_i=\bm{\eta}_j$ holds for all $i,j\in\mathcal{V}$.
Thus, the dual problem of \eqref{multi8} can be formulated as follows
	\begin{eqnarray}\label{dualprobp}
		\max\limits_{\substack{\bm{\eta}_i\\i=1,\ldots,n}}&& \sum_{i\in\mathcal{V}}h_i^{ls}(\bm{\eta}_i)\\
		\mbox{subject to }&&\bm{\eta}_i=\bm{\eta}_j,~~\forall i,j\in\mathcal{V}\nonumber
\end{eqnarray}\noindent where $h_i^{ls}(\bm{\eta}_i):=\inf_{y_i\in\mathcal{Y}_i,z_i\in\mathcal{Z}_i}\mathcal{L}_i^{ls}(y_i,z_i,\bm{\eta}_i)$. As in the previous section, a method based on the distributed sub-gradient approach is used to solve \eqref{dualprobp} in a distributed manner.
\subsubsection{Distributed algorithm}
Let us introduce $\phi_i^{(\nu)},\nu=1,\ldots,m+1$, $i\in\mathcal{V}$. Denoting $\bm{\eta}_i=[\eta_i^{(\nu)}]$ and $\bm{\phi}_i=[\phi_i^{(\nu)}],\nu=1,\ldots,m+1$, consider the following sequences for solving \eqref{multi8}
\begin{subequations}\label{alg_ls}
	\begin{eqnarray}
			[z_i(k+1),y_i(k+1)]=\inf_{y_i\in\mathcal{Y}_i,z_i\in\mathcal{Z}_i}\mathcal{L}_i^{ls}(y_i,z_i,\bm{\phi}_i(k)),\label{alg_ls0}
	\end{eqnarray}
\end{subequations}
where
\begin{subequations}[resume]
	\begin{eqnarray}
			\bm{\phi}_i(k+1)\!\!&=&\!\!\sum_{j\in\mathcal{V}}^{}a_{ij}\bm{\eta}_j(k+1),\label{alg_ls1}\\
			\bm{\eta}_i(k+1)\!\!&=&\!\!\bm{\phi}_i(k)+\alpha(k)\bm{\mathsf{g}}_i(z_i(k\!+\!1),y_i(k\!+\!1)),\label{alg_ls2}
	\end{eqnarray}
\end{subequations}
and $\bm{\mathsf{g}}_i(z_i,y_i)=[\mathsf{g}_i^{(1)}(z_i,y_i),\ldots,\mathsf{g}_i^{(m+1)}(z_i,y_i)]^T$. The proposed algorithm \eqref{alg_ls} can be implemented in a distributed manner under the assumption that $m$ is known to all agents. 
\subsubsection{Convergence Analysis}The following shows the convergence of \eqref{alg_ls}.
	\begin{thm}\label{thm:ls}
		Suppose that Assumptions  \ref{asm:graph}, \ref{asm:conv2} and \ref{asm:lsfeas} hold, and that $\alpha(k)$ satisfies $\sum_{k=0}^{\infty}\alpha(k)=\infty$ and $\sum_{k=0}^{\infty}\alpha^2(k)<\infty$. Then, the sequences \eqref{alg_ls}  solves optimization problem \eqref{multi8}, i.e., $\bm{\eta}_i(k)\rightarrow\thickbar{\bm{\eta}}$ and $(z_i(k),y_i(k))\rightarrow(\thickbar{z}_i,\thickbar{y}_i)$, as $k\rightarrow \infty$ for all $i\in\mathcal{V}$.
\end{thm}
See Appendix for the proof. When the desired priority list is changed, note that \eqref{alg_ls} requires only one tuning parameter $\kappa > \nabla D_i(y_i)$ for all $y_i\in\mathcal{Y}_i,i\in\mathcal{V}\backslash\mathcal{M}$, under assumption that the value of $m$ and $s_i$ are informed to all agents at every initial step of \eqref{alg_ls}. Given this, a question arises on developing a method such that $m$ and $s_i$ can be obtained by exchanging information between the neighbor agents; thus, the communication between the central controller and the agents can be reduced. The algorithm is discussed in the next section.

\section{DISTRIBUTED ENERGY MANAGEMENT : EDP + LOAD SHEDDING}
	This section presents how to integrate the proposed algorithms given in the previous sections. To this end, algorithms used to generate the value of  $y_{tot}$ and $m$ in a distributed manner are proposed. 

\subsection{Solving an overloading condition in the EDP}
In Section III, it is assumed that the EDP is feasible. However, there can be a case when there is no $\tilde{x}_i\in\mathcal{X}_i$ such that $\Upsilon(\tilde{\bm{x}})+\sum_{i\in \mathcal{V}}d_i -\sum_{i\in \mathcal{V}}\tilde{x}_i=0$ holds, meaning that an overloading condition occurs. To handle this, this section aims to discuss how to obtain $\sum_{i\in\mathcal{V}}s_i=y_{tot}$ in a distributed manner so that the feasibility of the EDP can be recovered.  
\subsubsection{A feasibility problem}
To find the necessary load shedding $y_{tot}=\sum_{i\in\mathcal{V}}s_i$, consider the following feasibility problem 
	\begin{eqnarray*}
		\begin{array}{cc}\mbox{\textbf{P1:}}~~~~~~&\begin{array}{c}
				\!\!\!\!\!\!\!\!\!\!\!\!\!\!\!\!\!\!\!\!\!\!\!\!\!\!\!\!\!\!\!\!\!\!\!\!\!\!\!\!\!\!\!\!\!\!\!\!\!\!\!\min\limits_{\substack{s_i\geq0,\bm{x}=[x_i],\\x_i\in\mathcal{X}_i,i\in\mathcal{V}}}~~ ~~~~~~ \sum_{i\in \mathcal{V}}^{}s_i^2\\
				\mbox{subject to~~~}\Upsilon(\bm{x})+\sum\limits_{i\in \mathcal{V}}\left(d_i -x_i-s_i\right)=0.
			\end{array}
		\end{array}
	\end{eqnarray*}
	Let us denote $s_i^*$ and $x_i^*$ as the optimal solution of this problem. Note that ${s}_i^*={s}_j^*$, $i,j\in\mathcal{V}$. Moreover, the problem yields the optimal sum of load shedding $y_{tot}=\sum_{i\in\mathcal{V}}{s}_i^*$ because $s_i^*$ is the smallest value such that $\sum_{i\in \mathcal{V}}{s}_i^*=\Upsilon(\bm{x}^*)+\sum_{i\in \mathcal{V}}\left(d_i -{x}_i^*\right)$ holds, $\bm{x}^*=[x_i^*], i\in\mathcal{V}$. 
\subsubsection{Relaxation of feasibility problem \textbf{P1}}
Consider the relaxed problem of \textbf{P1} defined as
\begin{eqnarray*}
	\begin{array}{cc}\mbox{\textbf{P2:}}~~&\begin{array}{c}
			\!\!\!\!\!\!\min\limits_{\substack{s_i\geq 0,\bm{u}=[u_i],\bm{x}=[x_i],\\x_i\in\mathcal{X}_i,i\in\mathcal{V}}}~~~~~ \sum_{i\in \mathcal{V}}^{}s_i^2+\tau\sum_{i\in \mathcal{V}}^{}x_i^2,\\
			\mbox{subject to~~~}\!\!\!\sum\limits_{i\in \mathcal{V}}\left(u_i^2+d_i -x_i-s_i\right)\leq\!0,\\
			R\bm{x}=\bm{u},
		\end{array}
	\end{array}
\end{eqnarray*}
where $\tau>0$ is a  tuning parameter.  Denoting $\thickbar{x}_i,\thickbar{u}_i$ and $\thickbar{s}_i$ as the optimal solution of this problem, it follows that $\thickbar{s}_i=\thickbar{s}_j$, $i,j\in\mathcal{V}$. 

{Note that \textbf{P2} is convex as the equality constraint in \textbf{P1} is relaxed into an inequality. Moreover, the term $\tau\sum_{i\in \mathcal{V}}^{}x_i^2$ is included as an additional objective function with $\tau>0$. Having this, the objective function of \textbf{P2} is monotonically increasing on the set $\{x_i,s_i|x_i\in\mbox{relint}(\mathcal{X}_i),s_i>0\}$}. Hence, Theorem \ref{thm:relax} can be re-established to show that the optimal solution of this problem meets equality $\sum_{i\in \mathcal{V}}\left(\thickbar{u}_i^2+d_i -\thickbar{x}_i-\thickbar{s}_i\right)=0$, i.e., $\sum_{i\in \mathcal{V}}\thickbar{s}_i=\Upsilon(\thickbar{\bm{x}}_i)+\sum_{i\in \mathcal{V}}\left(d_i -\thickbar{x}_i\right)$ where $\thickbar{\bm{x}}_i=[\thickbar{x}_i],i\in\mathcal{V}$. However, due to the additional term $\tau\sum_{i\in \mathcal{V}}^{}x_i^2$, we have $\sum_{i\in\mathcal{V}}\thickbar{s}_i\geq\sum_{i\in\mathcal{V}}s_i^*$, meaning that the optimal solution of \textbf{P2} is not the optimal solution of \textbf{P1}. Nevertheless, small value of $\|s_i^*-\thickbar{s}_i\|$ can be obtained by setting $\tau>0$ sufficiently small.  With some minor modifications, the distributed algorithm given in \eqref{alg_edp} can be adopted for obtaining $\thickbar{s}_i$ using only local information. The details are omitted due to page limit.

\subsection{Supporting algorithm for obtaining $m$}
\subsubsection{Review of a consensus algorithm}
At first, a consensus algorithm studied in \cite{Olshevsky2017} is reviewed. 
	Suppose that each $i$th agent has two local variables $\theta_i,\omega_i\in \mathbb{R}$ and applies a consensus algorithm given by
	\begin{subequations}\label{alg_est_m}
		\begin{eqnarray}
			\!\!\!\!\!\!	\!\!\!\theta_i(k+1)\!\!\!&=&\!\!\!\omega_i(k)+\frac{1}{2}\sum_{j\in\mathcal{N}_i} \frac{\omega_j(k)-\omega_i(k)}{\max\{|\mathcal{N}|_i,|\mathcal{N}_j|\}},\\
			\!\!\!	\!\!\!\!\!\!\omega_i(k+1)\!\!\!&=&\!\!\!\theta_i(k)\!+\!\left(\!1\!-\!\frac{2}{9n+1}\!\right)\!\!(\theta_i(k+1)-\theta_i(k)).
		\end{eqnarray}
	\end{subequations}
	Then, the following lemma holds.
\begin{lem}\cite[Theorem 2.1]{Olshevsky2017}\label{lem:linear_cons}
	Suppose that Assumption \ref{asm:graph} holds and that each node in an undirected graph $\mathcal{G}$ implements \eqref{alg_est_m} with $\omega_i(0)=\theta_i(0)$. Then, $\lim_{k\rightarrow\infty}\theta_i(k)=\frac{1}{n}\sum_{i=1}^{n}\omega_i(0)=:\theta_{ss}$ and 
	\begin{eqnarray*}
		\|\bm{\theta}(k)-\theta_{ss}\bm{1}\|_2^2\leq 2\left(1-\frac{1}{9n}\right)^{k-1}\|\bm{\theta}(0)-\theta_{ss}\bm{1}\|_2^2
	\end{eqnarray*}
	where $\bm{\theta}=[\theta_i],i=1,\ldots,n.$
\end{lem}

Lemma \ref{lem:linear_cons} implies that \eqref{alg_est_m} converges to the average of the initial values as $k\rightarrow\infty$. Moreover, $\|\bm{\theta}(k)-\theta_{ss}\bm{1}\|_{\infty}\leq \epsilon$ after $O(n \ln(\|\bm{\theta}(0)-\theta_{ss}\bm{1}\|/\epsilon))$ iterations. The following section describes the approach based on \eqref{alg_est_m} for obtaining the value of $m$ in a distributed manner.

\subsubsection{Applying consensus algorithm \eqref{alg_est_m} to generate $m$}
Let us consider the following value for the initial condition of \eqref{alg_est_m}
	\begin{eqnarray}\label{eq_initial}
		\omega_i(0)&=& \left\{ \begin{array}{cl}
			{np_i}/{|M_\ell|}
			& i\in\mathcal{M}_{\ell},\ell\in[1,m],\\
			{n}/{|\mathcal{V}\backslash\mathcal{M}|} &  i\in\mathcal{V}\backslash\mathcal{M},
		\end{array}\right.
	\end{eqnarray}
	with $\theta_i(0)=\omega_i(0)$. Then, $\lim_{k\rightarrow\infty}\theta_i(k)=\theta_{ss}=1+\sum_{\ell=1}^{m}\ell=(2+{m(m+1)})/{2}$, $i\in\mathcal{V}$.
	It is easy to see that $\theta_{ss}$ is always unique for any $m\in[0,n]$. For example, $\theta_{ss}=1$ when $m=0$, $\theta_{ss}=2$ when $m=1$, and so on. Since $m$ is upper-bounded by $n$, the consensus value $\theta_{ss}$ can be decoded to obtain the value of $m$.

In practice, it may be hard to obtain $\theta_i=\theta_{ss}$ since it holds asymptotically. Thus, we may need to use $\theta_i(k)$ instead of $\theta_{ss}$. From Lemma \ref{lem:linear_cons}, $\theta_i(k)=\theta_{ss}\pm \epsilon$ for all $k\geq K_n(\epsilon)$ where
$K_n(\epsilon):=O(n \ln(\|\bm{\theta}(0)-\theta_{ss}\bm{1}\|/\epsilon))$. The following describes the value of $\theta_i(k)\!\!=\!\theta_{ss}\!\pm\!\epsilon$ for given different values of $m$.
\begin{table}[H]
	\caption{$\theta_i(k)\!\!=\!\theta_{ss}\!\pm\!\epsilon$ for given different value of $m$.}\label{tab:encode}
	\begin{center}
		\begin{tabular}{ | m{2.6cm} | m{0.6cm}| m{0.6cm} | m{0.45cm} |m{2.6cm} | } 
			\hline
			$\theta_i(k)=\theta_{ss}\pm\epsilon $&$1\pm \epsilon$  &$2\pm \epsilon$   &$\ldots$& $(2+n^2+n\pm 2\epsilon)/{2}$  \\
			\hline
			$m$&$~~$0$~~$ & $~~$1$~~$&$\ldots$&$~~~~~~~~~n~~~~~~~~$\\
			\hline
		\end{tabular}
	\end{center}                
\end{table}
\noindent
In view of Table \ref{tab:encode}, choosing $\epsilon$ satisfying $0<\epsilon<0.5$ yields an unique set $\theta_i(k)\in[\theta_{ss}-\epsilon,\theta_{ss}+\epsilon]$ for any $m$, $m\in[0,n]$. Having this, 
the value of $m$ can be still attained by running the algorithm \eqref{alg_est_m} for at least $K_n(\epsilon)$ iterations where $\epsilon\in(0,0.5)$.
\begin{rem}
	The proposed algorithm \eqref{alg_est_m} can be implemented under the assumption that the $i$th agent has knowledge of the cardinality of the set $\mathcal{M}_\ell$ (if $i\in \mathcal{M}_\ell$) or $\mathcal{V}\backslash\mathcal{M}$ (if $i\in \mathcal{V}\backslash\mathcal{M}$), $\ell\in\{1,\ldots,m\}$. In practice, this information can be obtained from the central controller. Note that when the desired priority list is changed, then, the central controller broadcast the new value of $p_i$ and $|\mathcal{M}_\ell|$ only to the agents that are affected by the modification. In light of this, the method requires less communication between the agents and the central controller, compared to the case when $m$ is broadcast to all agents.
		Moreover, from the initial condition \eqref{eq_initial}, the information regarding the choice of the load shedding program of the $i$th customer is not disclosed to other agents, meaning that the privacy is preserved. 
\end{rem}

\subsection{Proposed distributed energy management}
The integration of the proposed EDP and the load shedding method is formally written in \textbf{Algorithm 1}. Bear in mind that the required demand for the EDP is updated in \textbf{Algorithm 1} by considering the load shedding, i.e., $d_i\leftarrow d_i-s_i$ in Stage 2. 
\begin{algorithm}
	\small
	\caption{Distributed EDP + load shedding}
	\label{alg1}
		\vspace{1mm}
		\textbf{Input from the operator:} $p_i$, $|\mathcal{M}_{\ell}|$ if $i\in\mathcal{M}_{\ell}$, $|\mathcal{V}\backslash\mathcal{M}|$ otherwise\\
		\textbf{Stage 1a : Finding $m$ }\hrulefill
		\begin{algorithmic}
			\State Initialize $\theta_{i}(0)=\omega_i(0)$ using \eqref{eq_initial}
			\State Run \eqref{alg_est_m} for $k\geq K_n(\epsilon)$
			\State \textbf{Output} : All agent know $m$ by using Table \ref{tab:encode}
		\end{algorithmic}
		\textbf{Stage 1b : Estimating $y_{tot}$ }\hrulefill
		\begin{algorithmic}
			\State Solve \textbf{P2} in a distributed manner
			\State \textbf{Output} : Each agent holds $s_i$
		\end{algorithmic}
		\textbf{Stage 2: Determining $x_i,y_i$}\hrulefill
		\begin{algorithmic}
			\State \textbf{Input} :  $m,s_i$
			\State Set $d_i\leftarrow d_i-s_i$
			\State Initialize $\lambda_i(0)\!\geq\! 0,\bm{\xi}_i(0)\in\mathbb{R}^{n}$, $\bm{\eta}_i(0)\in\mathbb{R}^{m+1}$
			\Repeat {$~~~k=0,1,\ldots$} 
			\State Exchange $\lambda_i(k),\bm{\xi}_i(k),\bm{\eta}_i(k)$ to neighboring agents. 
			\State Run \eqref{alg_edp0}-\eqref{alg_edp1} and \eqref{alg_ls1}
			\State Calculates $x_i(k+1)$ and $y_i(k+1)$ using \eqref{alg_edp_prim} and \eqref{alg_ls0}
			\State Update $\lambda_i(k+1),\bm{\xi}_i(k+1),\bm{\eta}_i(k)$ using \eqref{alg_edp3}-\eqref{alg_edp4} and \eqref{alg_ls2}
			\Until{converge}
	\end{algorithmic}
\end{algorithm} 
In the next section, numerical simulations are used to demonstrate that the proposed algorithms yield the desired output of the EDP and the priority-considered load shedding.

\section{Numerical Simulations}
Note that \textbf{Stage 2} of \textbf{Algorithm 1} plays an essential role in determining the dispatched power and load shedding at each bus. Moreover, \textbf{Stages} \textbf{1a} and \textbf{1b} can be seen as two supporting algorithms.  Having this, due to page limit, this section only presents the simulation results for the algorithms under \textbf{Stage 2} of \textbf{Algorithm 1}, i.e., the methods given in \eqref{alg_edp} and \eqref{alg_ls} are validated. 
For this purpose, consider the  IEEE-30 bus \cite{Wang2008}. The system consists of 6 generators and 24 load buses. Furthermore, it is assumed that the communication between the buses is modeled as an undirected graph $\mathcal{G}=(\mathcal{V},\mathcal{E})$, $\mathcal{V}=\{1,\ldots,30\}$ where the communication paths do not necessarily coincide with the power flow.
\subsection{The proposed EDP considering transmission losses}
\begin{table}
	\begin{center}
		\caption{Generators parameters ($MU$=money unit).}\label{tab:par}
		\begin{tabular}{||c c c c||} 
			\hline
			Generator & $\mathsf{a}_i$ $(MU/MW^2
			)$ & $\mathsf{b}_i$ $(MU/MW
			)$  & $x_i^{max}(MW)$ \\ [0.5ex] 
			\hline\hline
			1 & 0.08 & 2 & 20 \\ 
			\hline
			2 & 0.06 & 3 & 10 \\
			\hline
			3 & 0.07 & 4 & 30 \\
			\hline
			4 & 0.06 & 4 & 15 \\
			\hline
			5 & 0.08 & 2.5 & 10\\
			\hline
			6 & 0.08 & 2.5 & 8 \\ 
			\hline
		\end{tabular}
	\end{center}
\end{table}
The simulation results of the proposed distributed EDP given in \eqref{alg_edp} is presented in this section. Table \ref{tab:par} describes $x_i^{max}$ and the parameters of the cost function of the $i$th generator with $C_i(x_i)=\mathsf{a}_ix_i^2+\mathsf{b}_ix_i$. The minimum generation $x_i^{min}$ for all generators are 5MW. The $B$-matrix is obtained from \cite{Wang2008} and is given as follows:
\begin{eqnarray*}
	B\!=\!10^{-2}\!\!\mato{cccccc}{\!\!\!13.82\!&\!-2.99\!&\!0.44\!&\!-0.22\!&\!-0.10\!&\!-0.08\!\!\!\\
		\!\!\!-2.99\!&\!4.87\!&\!-0.25\!&\!0.04\!&\!0.16\! &\!0.41\!\\
		\!\!\!0.44\!&\!-0.25\!&\!1.82\!&\!-0.70\!&\!-0.66\!&\!-0.66\!\\
		\!\!\!-0.22\!&\!0.04\!&\!-0.70\!&\!1.37\!&\!0.50\!&\!0.33\!\\
		\!\!\!-0.10\!&\!0.16\!&\!-0.66\!&\!0.50\!&\!1.09\!&\!0.05\!\\
		\!\!\!-0.08\!&\!0.41\!&\!-0.66\!&\!0.33\!&\!0.05\!&\!2.44\!}
\end{eqnarray*}
At first, the proposed relaxed problem \eqref{edp2relmod} is solved in a centralized way to show that the optimal solution $\thickbar{x}_i$ yields $\Upsilon(\thickbar{\bm{x}})+\sum_{i\in\mathcal{V}}d_i-\sum_{i\in\mathcal{V}}\thickbar{x}_i=0$, $\thickbar{\bm{x}}=[\thickbar{x}_i]$. The results are depicted in Table \ref{tab:centr}. This implied that the supply-demand balance is attained using the relaxed-problem \eqref{edp2relmod}. 
\begin{table}
	\begin{center}
		\caption{Optimal solution of \eqref{edp2relmod} obtained using a centralized method. }\label{tab:centr}
		\begin{tabular}{||c c c||} 
			\hline
			$\sum_{i\in\mathcal{V}}d_i$ &  $\thickbar{\bm{x}}=[\thickbar{x}_i]$& $\Upsilon(\thickbar{\bm{x}})-\bm{1}^T\thickbar{\bm{x}}$\\ [0.5ex] 
			\hline\hline
			36&[5,~6.0836,~8.8734,~7.31,~8.2366,~6.57]&-36\\
			\hline
			48&[5,~7.406,~14.844,~11.544,~10,~8]&-48\\
			\hline
			55.2&[5,~9.4079,~19.5281,~15,~10,~8]&-55.2\\\hline
		\end{tabular}
	\end{center}
\end{table}

For the purpose of solving \eqref{edp2relmod} in a distributed setting, the algorithm \eqref{alg_edp} is employed with $\alpha_i(k)={100}/{k^{0.6}}$ and the load demand is 2MW at all load buses, meaning that $\sum_{i\in\mathcal{V}}d_i=48$ since there are 24 load buses. Figure \ref{fig:edpres1} shows that the total power
mismatch converges to zero as iteration increases.  Moreover, Figure \ref{fig:edpres2} describes the dispatched power at each generator. Note that the optimal solution from the distributed algorithm converges to the solution by a centralized method shown in Table \ref{tab:centr}. Moreover, Table \ref{tab:distr} shows the dispatched power obtained using the proposed distributed method \eqref{alg_edp} for given different load demands. Observe that the results are the same as the solution obtained using a centralized approach described in Table \ref{tab:centr}. 
\begin{figure}
	\centering
	\includegraphics[height=3.5cm]{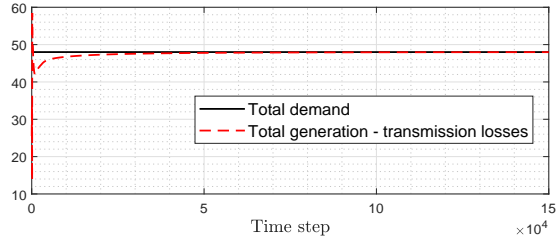}
	\caption{The dispatched power converges to the value such that $\sum_{i\in\mathcal{V}}d_i=\sum_{i\in\mathcal{V}}x_i-\Upsilon(\bm{x})$, $\bm{x}=[x_i]$.}
	\label{fig:edpres1}
\end{figure}%
\begin{figure}
	\centering
	\includegraphics[height=6cm]{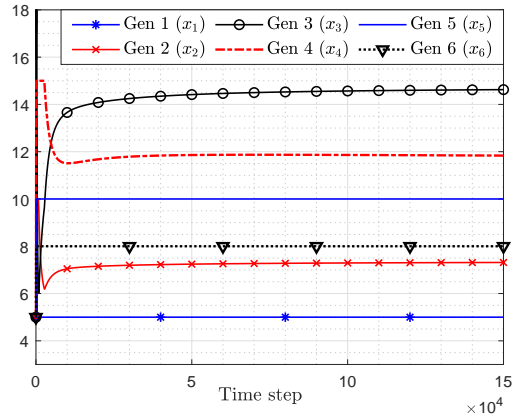}
	\caption{For the given load demand $\sum_{i\in\mathcal{V}}d_i=48$, the dispatched power $x_i(k)$ converges to $[5,7.38,14.78,11.64,10,8]$. }
	\label{fig:edpres2}
\end{figure}%
\begin{table}
	\begin{center}
		\caption{Optimal solution of \eqref{edp2relmod} obtained using the proposed distributed algorithm \eqref{alg_edp}. }\label{tab:distr}
		\begin{tabular}{||c c c||} 
			\hline
			$\sum_{i\in\mathcal{V}}d_i$ &  $\thickbar{\bm{x}}=[\thickbar{x}_i]$& $\Upsilon(\thickbar{\bm{x}})-\bm{1}^T\thickbar{\bm{x}}$\\ [0.5ex] 
			\hline\hline
			36&[5,~6.05,~8.82,~7.34,~8.2323,~6.6437]&-35.9999\\
			\hline
			48&[5,~7.38,~14.78,~11.64,~10,~8]&-48.0035\\
			\hline
			55.2&[5,~9.407,~19.5321,~15,~10,~8]&-55.2027\\\hline
		\end{tabular}
	\end{center}
\end{table}
\subsection{The proposed priority considered load shedding}
Numerical simulations are carried out to investigate the effectiveness of distributed algorithm \eqref{alg_ls} to solve \textbf{LS Problem} defined in Section II. The algorithm is applied with $\alpha(k)=1000/(k+500)$ to allocate the load shedding for 24 buses, i.e., $i\in\{7,\ldots,30\}$. 
In this study, $m=3$ is  considered where $\mathcal{M}_1=\{7\}$, $\mathcal{M}_2=\{8,9\}$, and $\mathcal{M}_3=\{10\}$. This means that we would like to shed the load in the $7$th, $8$th, $9$th, and $10$th bus first before shedding the remaining 20 buses. Moreover, the bus belonging to the set $\mathcal{M}_1$ has the highest, and $\mathcal{M}_3$ has the lowest priority to be shed.
Consider $\mathsf{r}_i=3$ and $D_i(y_i)=\frac{1}{2}\mathsf{q}_iy_i^2$ where $\mathsf{q}_i=1$ for all $i=\{11,12,21,22\}$, $\mathsf{q}_i=2$ for all $i=\{13,14,23,24\}$, $\mathsf{q}_i=2.5$ for all $i=\{15,16,25,26\}$, $\mathsf{q}_i=3$ for all $i=\{17,18,27,28\}$, and $\mathsf{q}_i=4$ for all $i=\{19,20,29,30\}$.

The algorithm is employed by setting $\kappa=40$ and $y_i^{max}=1.2$, $i\in\mathcal{V}$. At first, let us consider the case when $\sum_{i\in\mathcal{V}}s_i=y_{tot}=1.8$. Figure \ref{fig:ls1} shows that $y_7=y_7^{max}=1.2$ and $y_8=y_9=0.3$. Since $y_i\approx 0$ for all $i\in\mathcal{V}\backslash\{7,8,9\}$, we have $\sum_{i\in\mathcal{V}}y_i=y_{tot}$. Note that $y_8=y_9=0.3$ because the load shedding at the bus belonging to the set $\mathcal{M}_1$ is not enough to meet the required load shedding $y_{tot}$. Owing to this, the priority-considered load shedding scenario is achieved. 

Next, simulations are carried out by considering three different values of $y_{tot}$. The results are shown in Figure \ref{fig:ls2}. When, $y_{tot}=1$, only the bus belonging to the set $\mathcal{M}_1$ is shed. In the case $y_{tot}=4$, we have $y_7=y_{8}=y_{9}=1.2$ and $y_{10}=0.4$. Moreover, when $y_{tot}=6$, we have $y_7=y_{8}=y_{9}=y_{10}=\sum_{i\in V\backslash\{7,8,9,10\}}y_i=1.2$. Owing to this, the proposed method is able to yield the load shedding according to the desired priority list. Figure \ref{fig:ls3} shows the convergence of the algorithm in the case when $y_{tot}=6$. Note that $y_i$ converges to the value such that its summation over all $i\in\mathcal{V}$ coincides with the desired value $\sum_{i\in\mathcal{V}}s_i=y_{tot}=6$.
\begin{figure}
	\centering
	\includegraphics[height=3.5cm]{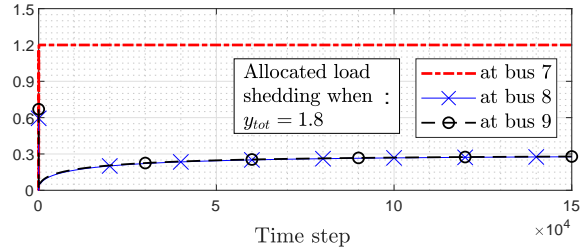}
	\caption{Let  $\mathcal{M}_1=\{7\},\mathcal{M}_2=\{8,9\}$, $y_7^{max}=1.2$, and $y_{tot}=1.8$, simulation results show that the desired priority-considered load shedding is obtained. }
	\label{fig:ls1}
\end{figure}%
\begin{figure}
	\centering
	\includegraphics[height=6cm]{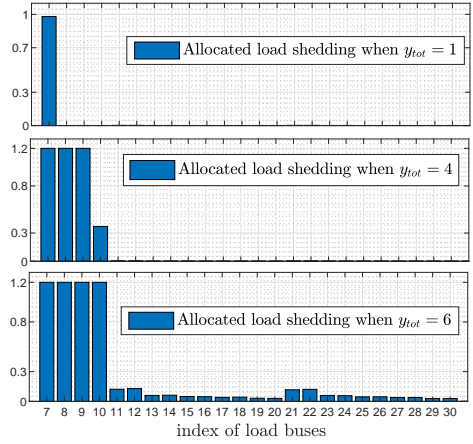}
	\caption{Simulation results for different $y_{tot}:1,4,$ and $6$. Given that $\mathcal{M}_1=\{7\},\mathcal{M}_2=\{8,9\}$, $\mathcal{M}_3=\{10\}$, and $y_i^{max}=1.2$ for all $i\in\mathcal{V},$ the scheduled load shedding follows the desired priority list.}
	\label{fig:ls2}
\end{figure}%
\begin{figure}
	\centering
	\includegraphics[height=3cm]{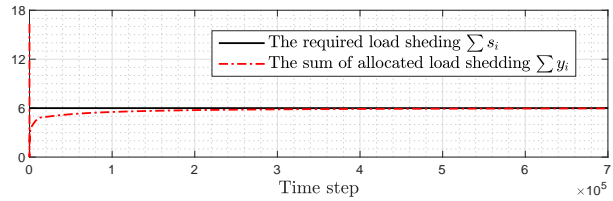}
	\caption{For the case when $\sum_{i\in\mathcal{V}}s_i=y_{tot}=6$ the value of $\sum_{i\in\mathcal{V}}y_i-\sum_{i\in\mathcal{V}}s_i$ converges to zero as the iterations increase.}
	\label{fig:ls3}
\end{figure}%

\section{Conclusion}
This paper formulates an optimization problem for the EDP to deal with the transmission loss power modeled in a quadratic function. In addition, an optimization problem is proposed to consider some priorities in the load shedding program. Under some assumptions, distributed algorithms for both problems are proposed. Moreover, a feasibility problem is adopted to find the necessary load shedding for handling an overloading condition in a distributed manner. 

\bibliographystyle{IEEEtran}
\bibliography{reference_fix}

\begin{thebibliography}{10}
\providecommand{\url}[1]{#1}
\csname url@samestyle\endcsname
\providecommand{\newblock}{\relax}
\providecommand{\bibinfo}[2]{#2}
\providecommand{\BIBentrySTDinterwordspacing}{\spaceskip=0pt\relax}
\providecommand{\BIBentryALTinterwordstretchfactor}{4}
\providecommand{\BIBentryALTinterwordspacing}{\spaceskip=\fontdimen2\font plus
\BIBentryALTinterwordstretchfactor\fontdimen3\font minus
  \fontdimen4\font\relax}
\providecommand{\BIBforeignlanguage}[2]{{%
\expandafter\ifx\csname l@#1\endcsname\relax
\typeout{** WARNING: IEEEtran.bst: No hyphenation pattern has been}%
\typeout{** loaded for the language `#1'. Using the pattern for}%
\typeout{** the default language instead.}%
\else
\language=\csname l@#1\endcsname
\fi
#2}}
\providecommand{\BIBdecl}{\relax}
\BIBdecl

\bibitem{fan_zhang_1998}
{Ji-Yuan Fan} and {Lan Zhang}, ``Real-time economic dispatch with line flow and
  emission constraints using quadratic programming,'' \emph{IEEE Transactions
  on Power Systems}, vol.~13, no.~2, pp. 320--325, May 1998.

\bibitem{Qiang2016}
{Qiang Wan}, W.~{Zhang}, Y.~{Xu}, and I.~{Khan}, ``Distributed control for
  energy management in a microgrid,'' in \emph{2016 IEEE/PES Transmission and
  Distribution Conference and Exposition (T\&D)}, Dallas, TX, 2016, pp. 1--5.

\bibitem{Yang2019}
T.~Yang, D.~Wu, H.~Fang, W.~Ren, H.~Wang, Y.~Hong, and K.~H. Johansson,
  ``Distributed energy resource coordination over time-varying directed
  communication networks,'' \emph{IEEE Transactions on Control of Network
  Systems}, vol.~6, no.~3, pp. 1124--1134, 2019.

\bibitem{Lin2019}
F.~{Lin}, ``Worst-case load shedding in electric power networks,'' \emph{IEEE
  Transactions on Control of Network Systems}, vol.~6, no.~3, pp. 1269--1277,
  2019.

\bibitem{Zhang2012}
Z.~{Zhang} and M.~{Chow}, ``Convergence analysis of the incremental cost
  consensus algorithm under different communication network topologies in a
  smart grid,'' \emph{IEEE Transactions on Power Systems}, vol.~27, no.~4, pp.
  1761--1768, 2012.

\bibitem{Kar2012}
S.~{Kar} and G.~{Hug}, ``Distributed robust economic dispatch in power systems:
  A consensus + innovations approach,'' in \emph{2012 IEEE Power and Energy
  Society General Meeting}, San Diego, CA, 2012, pp. 1--8.

\bibitem{Yang2017}
Z.~{Yang}, J.~{Xiang}, and Y.~{Li}, ``Distributed consensus based
  supply–demand balance algorithm for economic dispatch problem in a smart
  grid with switching graph,'' \emph{IEEE Transactions on Industrial
  Electronics}, vol.~64, no.~2, pp. 1600--1610, 2017.

\bibitem{Wang2019}
R.~{Wang}, Q.~{Li}, B.~{Zhang}, and L.~{Wang}, ``Distributed consensus based
  algorithm for economic dispatch in a microgrid,'' \emph{IEEE Transactions on
  Smart Grid}, vol.~10, no.~4, pp. 3630--3640, 2019.

\bibitem{Doan2017}
T.~T. {Doan} and C.~L. {Beck}, ``Distributed lagrangian methods for network
  resource allocation,'' in \emph{2017 IEEE Conference on Control Technology
  and Applications (CCTA)}, Mauna Lani, HI, 2017, pp. 650--655.

\bibitem{Yangvarying2017}
T.~{Yang}, J.~{Lu}, D.~{Wu}, J.~{Wu}, G.~{Shi}, Z.~{Meng}, and K.~H.
  {Johansson}, ``A distributed algorithm for economic dispatch over
  time-varying directed networks with delays,'' \emph{IEEE Transactions on
  Industrial Electronics}, vol.~64, no.~6, pp. 5095--5106, 2017.

\bibitem{Xu2019}
J.~Xu, S.~Zhu, Y.~C. Soh, and L.~Xie, ``A dual splitting approach for
  distributed resource allocation with regularization,'' \emph{IEEE
  Transactions on Control of Network Systems}, vol.~6, no.~1, pp. 403--414,
  2019.

\bibitem{Nedic2018ImprovedCR}
A.~{Nedić}, A.~{Olshevsky}, and W.~{Shi}, ``Improved convergence rates for
  distributed resource allocation,'' in \emph{2018 IEEE Conference on Decision
  and Control (CDC)}, Miami Beach, FL, 2018, pp. 172--177.

\bibitem{Binetti2014}
G.~{Binetti}, A.~{Davoudi}, F.~L. {Lewis}, D.~{Naso}, and B.~{Turchiano},
  ``Distributed consensus-based economic dispatch with transmission losses,''
  \emph{IEEE Transactions on Power Systems}, vol.~29, no.~4, pp. 1711--1720,
  2014.

\bibitem{Kim2019}
K.~Kim, ``Distributed learning algorithms and lossless convex relaxation for
  economic dispatch with transmission losses and capacity limits,''
  \emph{Mathematical Problems in Engineering}, vol. 2019, pp. 1--11, 2019.

\bibitem{Xing2015}
H.~{Xing}, Y.~{Mou}, M.~{Fu}, and Z.~{Lin}, ``Distributed algorithm for
  economic power dispatch including transmission losses,'' in \emph{2015
  European Control Conference (ECC)}, 2015, pp. 1076--1081.

\bibitem{Zhao2017}
C.~{Zhao}, J.~{He}, P.~{Cheng}, and J.~{Chen}, ``Consensus-based energy
  management in smart grid with transmission losses and directed
  communication,'' \emph{IEEE Transactions on Smart Grid}, vol.~8, no.~5, pp.
  2049--2061, Sep. 2017.

\bibitem{Lee2019}
S.~{Lee} and H.~{Shim}, ``Distributed algorithm for economic dispatch problem
  with separable losses,'' \emph{IEEE Control Systems Letters}, vol.~3, no.~3,
  pp. 685--690, July 2019.

\bibitem{Faranda2007}
R.~{Faranda}, A.~{Pievatolo}, and E.~{Tironi}, ``Load shedding: A new
  proposal,'' \emph{IEEE Transactions on Power Systems}, vol.~22, no.~4, pp.
  2086--2093, 2007.

\bibitem{Hussain2020}
A.~{Hussain}, V.~{Bui}, and H.~{Kim}, ``An effort-based reward approach for
  allocating load shedding amount in networked microgrids using multiagent
  system,'' \emph{IEEE Transactions on Industrial Informatics}, vol.~16, no.~4,
  pp. 2268--2279, 2020.

\bibitem{Kato2014}
T.~Kato, H.~Takahashi, K.~Sasai, G.~Kitagata, H.-M. Kim, and T.~Kinoshita,
  ``Priority-based hierarchical operational management for multiagent-based
  microgrids,'' \emph{Energies}, vol.~7, no.~4, p. 2051–2078, Mar 2014.

\bibitem{Hussain2018}
A.~{Hussain}, V.~{Bui}, and H.~{Kim}, ``A resilient and privacy-preserving
  energy management strategy for networked microgrids,'' \emph{IEEE
  Transactions on Smart Grid}, vol.~9, no.~3, pp. 2127--2139, 2018.

\bibitem{Grewal1998}
G.~S. {Grewal}, J.~W. {Konowalec}, and M.~{Hakim}, ``Optimization of a load
  shedding scheme,'' \emph{IEEE Industry Applications Magazine}, vol.~4, no.~4,
  pp. 25--30, 1998.

\bibitem{Teshome2015}
D.~F. Teshome, P.~F. Correia, and K.~L. Lian, ``{Stochastic Optimization for
  Network-Constrained Power System Scheduling Problem},'' \emph{Mathematical
  Problems in Engineering}, vol. 2015, 2015.

\bibitem{lewis_2014}
F.~L. Lewis, H.~Zhang, and K.~Hengster-Movric, \emph{Cooperative Control of
  Multi-Agent Systems: Optimal and Adaptive Design Approaches}.\hskip 1em plus
  0.5em minus 0.4em\relax Springer-Verlag, 2014.

\bibitem{grainger_stevenson_chang_2016}
J.~J. Grainger and W.~D. Stevenson, \emph{Power System Analysis}.\hskip 1em
  plus 0.5em minus 0.4em\relax New York: McGraw-Hill, 1994.

\bibitem{Nali2011}
N.~{Li}, L.~{Chen}, and S.~H. {Low}, ``Optimal demand response based on utility
  maximization in power networks,'' in \emph{2011 IEEE Power and Energy Society
  General Meeting}, Detroit, MI, USA, 2011.

\bibitem{Bertsekas1989}
D.~P. Bertsekas and J.~N. Tsitsiklis, \emph{{Parallel and Distributed
  Computation: Numerical Methods}}.\hskip 1em plus 0.5em minus 0.4em\relax
  Englewood Cliffs, New Jersey: Prentice-Hall, 1989.

\bibitem{Nedic2010}
A.~{Nedic}, A.~{Ozdaglar}, and P.~A. {Parrilo}, ``Constrained consensus and
  optimization in multi-agent networks,'' \emph{IEEE Transactions on Automatic
  Control}, vol.~55, no.~4, pp. 922--938, 2010.

\bibitem{Olshevsky2017}
A.~Olshevsky, ``Linear time average consensus and distributed optimization on
  fixed graphs,'' \emph{SIAM Journal on Control and Optimization}, vol.~55,
  no.~6, pp. 3990--4014, 2017.

\bibitem{Wang2008}
L.~Wang and C.~Singh, ``Balancing risk and cost in fuzzy economic dispatch
  including wind power penetration based on particle swarm optimization,''
  \emph{Electric Power Systems Research}, vol.~78, no.~8, pp. 1361 -- 1368,
  2008.

\bibitem{Nedic2009}
A.~Nedić and A.~Ozdaglar, ``Approximate primal solutions and rate analysis for
  dual subgradient methods,'' \emph{SIAM Journal on Optimization}, vol.~19,
  no.~4, pp. 1757--1780, 2009.

\end{thebibliography}
\appendix
\subsection{Proof of Theorem \ref{thm:relax}}\label{append:thmrelax}
Note that $\mathcal{L}(\bm{x},\bm{u},\lambda,\bm{\xi})$ given in \eqref{globallagr} is the relaxed Lagrangian function of \eqref{edp2rel} since the constraint $x_i\in\mathcal{X}_i$ is considered implicitly. For the purpose of showing \eqref{kkt}, let us consider the full Lagrangian function of \eqref{edp2rel} as follows:
\begin{eqnarray*}
	\hat{\mathcal{L}}(\bm{x},\bm{u},\lambda,\bm{\xi},\bm{\varphi},\bm{\sigma})\!\!\!&:=&\!\!\!	\sum_{i\in\mathcal{V}}\left(\sigma_i(x_i\!-\!x_i^{max})+\varphi_i(x_i^{min}-x_i)\right)+\mathcal{L}(\bm{x},\bm{u},\lambda,\bm{\xi})
\end{eqnarray*}
where $\bm{\varphi}=[\varphi_i],\bm{\sigma}=[\sigma_i],$$i\in\mathcal{V}$. The following is complementary slackness condition corresponding to inequality constraint \eqref{eqbalrel}
\begin{eqnarray}\label{slackness}
	\thickbar{\lambda}\geq 0,\quad\thickbar{\lambda} \left(\sum\limits_{i\in \mathcal{V}}(\thickbar{u}_i^2+d_i -\thickbar{x}_i)\right)=0.
\end{eqnarray}
Let us consider the case when $\thickbar{\lambda}=0$ and check the existence of  $\thickbar{x}_i$, $\thickbar{u}_i$, $\thickbar{\sigma}_i$, and $\thickbar{\varphi}_i, i\in\mathcal{V},$ satisfying the KKT conditions. At first, $\partial\hat{\mathcal{L}}/\partial x_i=0$ and $ \partial\hat{\mathcal{L}}/\partial u_i=0$ can be written as follows:
\begin{eqnarray}
	\!\!\!\partial\hat{\mathcal{L}}/\partial x_i&=&\nabla C_i(\thickbar{x}_i)\!-\!\thickbar{\lambda}+\!\sum_{j\in\mathcal{V}}\thickbar{\xi}^{(j)}r_{ji}\!+\!\thickbar{\sigma}_i\!-\!\thickbar{\varphi}_i\!=\!0,\label{derif1}\\
	\!\!\!\partial\hat{\mathcal{L}}/\partial u_i&=&2\thickbar{\lambda}u_i-\thickbar{\xi}^{(i)}=0.\label{derif2}
\end{eqnarray}
Since $\thickbar{\lambda}=0$, from \eqref{derif2} we have $\thickbar{\xi}^{(j)}=0$ for all $j\in\mathcal{V}$. Substituting it to \eqref{derif1} yields $\nabla C_i(\thickbar{x}_i)+\thickbar{\sigma}_i-\thickbar{\varphi}_i=0$.

Having this, let us consider the case where $\thickbar{\varphi}_i=0,\thickbar{\sigma}_i=0$, $i\in\mathcal{V}$. This means $\thickbar{x}_i\in \mbox{relint}(\mathcal{X}_i)$ and yields $\nabla C_i(\thickbar{x_i})=0$. Note that this contradicts  the assumption that $\nabla C_i(x_i)>0$ for all $x_i\in \mbox{relint}(\mathcal{X}_i)$ described in Assumption \ref{asm:conv}. Next, consider the case where $\thickbar{\varphi}_i=0$ and $\thickbar{\sigma}_i\not=0,$ meaning that $\thickbar{x}_i=x_i^{max}$ for all $i\in\mathcal{V}$. It follows that $C_i(\thickbar{x}_i)+\thickbar{\sigma}_i=0$. 
Since $\thickbar{\sigma}_i>0$, it yields $\nabla C_i(\thickbar{x}_i^{max})<0$ which can not be true according to Assumption \ref{asm:conv}. Furthermore, let us consider the case where  $\thickbar{\varphi}_i\not=0$ and $\thickbar{\sigma}_i=0,$ meaning that $\thickbar{x}_i=x_i^{min}$, $i\in\mathcal{V}$. In this case, we have $\nabla C_i(\thickbar{x}_i^{min})-\thickbar{\varphi}_i=0$, which can be true. In conclusion, we have $\thickbar{x}_i=x_i^{min}$ when $\lambda=0$. By using \eqref{kkt3} and $B=R^TR$, we have $\sum_{i\in \mathcal{V}}\thickbar{u}_i^2=\Upsilon(\bm{x}_i^{min})$.The following holds according to Assumption \ref{asm:feas}
\begin{eqnarray*}
	\sum\limits_{i\in \mathcal{V}}\thickbar{u}_i^2+d_i-\thickbar{x}_i=\Upsilon(\bm{x}_i^{min})+\sum\limits_{i\in \mathcal{V}}\left(d_i-x_i^{min}\right)>0.
\end{eqnarray*}
Note that this contradicts the assumption that $\thickbar{\lambda}=0$ which equivalent to $\sum_{i\in \mathcal{V}}\thickbar{u}_i^2+d_i -\thickbar{x}_i<0$. 
Owing to this, $\thickbar{\lambda}>0$, meaning that equality constraint \eqref{kkt2} holds.  
\subsection{Proof of Claim 1}\label{append:claim}
Because $x_i(k+1)$ and $u_i(k+1)$ satisfy \eqref{alg_edp_prim}, i.e., $[x_i(k\!+\!1),u_i(k\!+\!1)]\!=\!\inf_{{x_i\in\mathcal{X}_i,u_i\in\mathcal{U}}}\!\mathcal{L}_i(x_i,u_i,v_i(k),\bm{w}_i(k))$ , for any $k\geq 0,i\in \mathcal{V}$, we have
\begin{eqnarray}\label{zero}
	0&\leq&-\mathcal{L}_i(x_i(k\!+\!1),u_i(k\!+\!1),v_i(k),\bm{w}_i(k))\nonumber\\&&+\mathcal{L}_i(\thickbar{x}_i,\thickbar{u}_i,v_i(k),\bm{w}_i(k)).
\end{eqnarray}
Due to the same reason, $q_i(v_i(k),\bm{w}_i(k))=-h_i(v_i(k),\bm{w}_i(k))=-\inf_{x_i\in\mathcal{X}_i,u_i\in\mathcal{U}}\mathcal{L}_i(x_i,u_i,v_i(k),\bm{w}_i(k))$ can be rewritten as
\begin{eqnarray}\label{eq1}
	&&\!\!\!\!\!\!\!\!\!\!\!\!\!\!\!\!\!\!\!\!\!q_i(v_i(k),\bm{w}_i(k))\!\!=\!-\mathcal{L}_i(x_i(k\!+\!1),\!{u}_i(k\!+\!1),v_i(k),\bm{w}_i(k)).
\end{eqnarray}
The zero-duality gap  implies that  $\sum_{i\in\mathcal{V}}C_i(\thickbar{x}_i)=\sum_{i\in\mathcal{V}}h_i(\thickbar{\lambda},\thickbar{\bm{\xi}})=-\sum_{i\in\mathcal{V}}q_i(\thickbar{\lambda},\thickbar{\bm{\xi}})$. As a result, we have
\begin{eqnarray}\label{eq2}
	&&\!\!\!\!\!\!\!\!\!\!\sum_{i\in \mathcal{V}}\mathcal{L}_i(\thickbar{x}_i,\thickbar{u}_i,v_i(k),\bm{w}_i(k))\!=\!\nonumber\\&& -\sum_{i\in \mathcal{V}}\!q_i(\thickbar{\lambda},\thickbar{\bm{\xi}})\!+\!\sum_{i\in \mathcal{V}}v_i(k)(\thickbar{u}_i^2\!+\!d_i\!-\!\thickbar{x}_i)\!\nonumber\\&&+\!\sum_{i\in \mathcal{V}}w_i^{(i)}(k)(r_{ii}\thickbar{x}_i\!-\!\thickbar{u}_i)+\!\sum_{i\in \mathcal{V}}\sum_{j\in\mathcal{V},j\not=i}\!\!\!\!\!\!w_i^{(j)}\!(k)r_{ji}\thickbar{x}_i,~~\label{two}
\end{eqnarray}
where $\bm{w}_i=[w_i^{(j)}]$, $j\in\mathcal{V}.$
The following is obtained from the summation of \eqref{zero} over all $i\in\mathcal{V}$.
\begin{eqnarray}\label{finala}
	\!\!\!\!\!\!\!\!\!\!\!&\!\!\!\!\!&\!\!\!\!\!\!\!\!\!\!\!\!\!\!\!\!0\leq-\sum_{i\in \mathcal{V}}\mathcal{L}_i(x_i(k\!+\!1),u_i(k\!+\!1),v_i(k),\bm{w}_i(k))\\&&+\sum_{i\in \mathcal{V}}\mathcal{L}_i(\thickbar{x}_i,\thickbar{u}_i,v_i(k),\bm{w}_i(k)).\nonumber
\end{eqnarray}
Substituting \eqref{eq1}-\eqref{eq2} into \eqref{finala} yields
\begin{eqnarray*}
	0&\leq&\sum_{i\in \mathcal{V}}\!q_i(v_i(k),\bm{w}_i(k))-\sum_{i\in \mathcal{V}}q_i(\thickbar{\lambda},\thickbar{\bm{\xi}})\\&&\!+\!\sum_{i\in \mathcal{V}}v_i(k)(\thickbar{u}_i^2\!+\!d_i\!-\!\thickbar{x}_i)
	\!+\!\sum_{i\in \mathcal{V}}w_i^{(i)}(k)(r_{ii}\thickbar{x}_i\!-\!\thickbar{u}_i)\\&&+\!\sum_{i\in \mathcal{V}}\sum_{j\in\mathcal{V},j\not=i}\!\!\!\!\!\!w_i^{(j)}\!(k)r_{ji}\thickbar{x}_i.\nonumber
\end{eqnarray*}
Since $q_i$ is a convex function, it follows that
\begin{eqnarray}\label{final}
	0&\leq&\sum_{i\in \mathcal{V}}\partial q_i(\thickbar{\lambda},\thickbar{\bm{\xi}})^T\left(\mato{c}{\!\!\thickbar{\lambda}\!\!\\\!\!\thickbar{\bm{\xi}}\!\!}\!-\!\mato{c}{	\!\!\!\!{v}_i(k)\!\!\!\!\\ \!\!\!\!\bm{w}_i(k)\!\!\!\!}\right)\!\nonumber\\&&+\!\sum_{i\in \mathcal{V}}v_i(k)(\thickbar{u}_i^2\!+\!d_i\!-\!\thickbar{x}_i)\!
	+\!\sum_{i\i \mathcal{V}}w_i^{(i)}(k)(r_{ii}\thickbar{x}_i\!-\!\thickbar{u}_i)\nonumber\\&&+\!\sum_{i\i \mathcal{V}}\sum_{j\in\mathcal{V},j\not=i}\!\!\!\!\!\!w_i^{(j)}\!(k)r_{ji}\thickbar{x}_i\nonumber
	\\&\leq& \sum_{i\in \mathcal{V}}\Gamma\left\|\mato{c}{\thickbar{\lambda}\\\thickbar{\bm{\xi}}}-\mato{c}{{v}_i(k)\\ \bm{w}_i(k)}\right\|\!\nonumber\\&&+\!\sum_{i\in \mathcal{V}}v_i(k)(\thickbar{u}_i^2\!+\!d_i\!-\!\thickbar{x}_i)\!+\!\sum_{i\in \mathcal{V}}w_i^{(i)}(k)(r_{ii}\thickbar{x}_i\!-\!\thickbar{u}_i)\nonumber\\&&+\!\sum_{i\in \mathcal{V}}\sum_{j\in\mathcal{V},j\not=i}\!\!\!\!\!\!w_i^{(j)}\!(k)r_{ji}\thickbar{x}_i.
\end{eqnarray}
Note that the right hand side of \eqref{final} goes to zero as $k\rightarrow \infty$ because $(v_i(k),\bm{w}_i(k))\!\!\rightarrow\!\!(\thickbar{\lambda},\thickbar{\bm{\xi}})$ as $k\rightarrow \infty$, $\sum_{i\in \mathcal{V}}(\thickbar{u}_i^2+d_i-\thickbar{x}_i)\leq0$, and $\sum_{i\in\mathcal{V}}(\sum_{j\in \mathcal{V}}r_{ji}\thickbar{x}_i-\thickbar{u}_i)=0$. This implies
$\lim_{k\rightarrow\infty}\!\sum_{i\in \mathcal{V}}\mathcal{L}_i(\thickbar{x}_i,\thickbar{u}_i,v_i(k),\bm{w}_i(k))\!-\!\sum_{i\in \mathcal{V}}\mathcal{L}_i(x_i(k\!+\!1),u_i(k\!+\!1),v_i(k),\bm{w}_i(k))\!=\!0$ which is equivalent to \eqref{sub1}.
\subsection{Proof of Theorem \ref{thm:ls}}\label{append:thmls}
Given that $y_i(k)\in\mathcal{Y}_i$ and $z_i(k)\in\mathcal{Z}_i$, there exists $\Lambda>0$ such $\|\bm{\mathsf{g}}_i(z_i(k),y_i(k))\|\leq \Lambda$ for all $k\geq 0$. Thus, $\bm{\eta}_i(k)\rightarrow\thickbar{\bm{\eta}}$ as $k\rightarrow\infty$ according to \cite[Proposition 5]{Nedic2009}. Using the similar reasoning for obtaining \eqref{final}, we have
\begin{eqnarray}\label{final2}
	0\!\!\!\!\!&\!\!\!\!\leq\!\!\!\!&\!\!\!\!\sum_{i\in \mathcal{V}}\mathcal{L}_i^{ls}(\thickbar{y}_i,\thickbar{z}_i,\bm{\phi}_i(k))\!-\!\sum_{i\in \mathcal{V}}\mathcal{L}_i^{ls}(y_i(k\!+\!1),z_i(k\!+\!1),\bm{\phi}_i(k))\nonumber\\&\leq&\sum_{i\in \mathcal{V}}\left(\Lambda\left\|\thickbar{\bm{\eta}}-\bm{\phi}_i(k)\right\|+\sum_{\nu=1}^{m+1}\phi_i^{(\nu)}(k)\mathsf{g}_i^{(\nu)}(\thickbar{z}_i,\thickbar{y}_i)\right).
\end{eqnarray}
Since $F_i$ is a strongly convex function, it follows that $\mathcal{L}_i^{ls}(y_i,z_i,\bm{\phi}_i(k))$ is a strongly convex function for any $\bm{\phi}_i(k)$. In other words, there exists $\gamma_i>0$ such that
\begin{eqnarray}\label{sub3}
	\!\!\!\!\!\!\!\!\!&\!\!\!\!\!\!\!\!\!&\!\!\!\!\!\!\!\!\!\!\!\!\frac{\gamma_i}{2}\left\|\mato{c}{\thickbar{y}_i\\\thickbar{z}_i}-\mato{c}{{y}_i(k+1)\\ z_i(k+1)}\right\|^2\leq\mathcal{L}^{ls}_i(\thickbar{y}_i,\thickbar{z}_i,\thickbar{\lambda},\bm{\phi}_i(k))\!\\
	&&\!\!\!\!\!\!-\mathcal{L}_i^{ls}(y_i(k\!+\!1),z_i(k\!+\!1),\bm{\phi}_i(k))\nonumber\\&&\!\!\!\!\!\!-\nabla\mathcal{L}_i^{ls}(y_i(k\!+\!1),z_i(k\!+\!1),\bm{\phi}_i(k))^T\!\!\left(\mato{c}{\!\!\!\thickbar{y}_i\!\!\!\\\!\!\!\thickbar{z}_i\!\!\!}\!-\!\mato{c}{\!\!\!\!y_i(k\!+\!1)\!\!\!\!\\ \!\!\!\!z_i(k\!+\!1)\!\!\!\!}\right)\!.\nonumber
\end{eqnarray}
By definition, $\nabla\mathcal{L}_i^{ls}(y_i(k\!+\!1),z_i(k\!+\!1),\bm{\phi}_i(k))^T=0$. Summing \eqref{sub3} over $i\in\mathcal{V}$ and using \eqref{final2} yields
\begin{eqnarray*}
	\sum_{i\in \mathcal{V}}\frac{\gamma_i}{2}\left\|\mato{c}{\thickbar{y}_i\\\thickbar{z}_i}-\mato{c}{{y}_i(k+1)\\ z_i(k+1)}\right\|^2\!\!\!&\!\!\!\leq\!\!\!&\!\!\! \sum_{i\in \mathcal{V}}\Lambda\left\|\thickbar{\bm{\eta}}-\bm{\phi}_i(k)\right\|\\\!\!\!\!\!\!\!\!\!\!\!\!\!\!\!\!\!\!\!\!\!&\!\!\!\!\!\!\!\!\!\!\!\!\!\!\!&\!\!\!\!\!\!\!\!\!\!\!\!\!\!\!\!\!\!\!\!\!+\sum_{i\in \mathcal{V}}\sum_{\nu=1}^{m+1}\phi_i^{(\nu)}(k)\mathsf{g}_i^{(\nu)}(\thickbar{z}_i,\thickbar{y}_i).
\end{eqnarray*}
Given that the right-hand side of this inequality goes to zero as $k\rightarrow \infty$ according to \eqref{final2}, we have $(z_i(k),y_i(k))\rightarrow(\thickbar{z}_i,\thickbar{y}_i)$ as $k\rightarrow \infty$.
\end{document}